\documentclass[a4paper,12pt]{article}
\usepackage[
    DIV=11,
    headinclude=false,
    footinclude=false,
]{typearea}

\usepackage[T1]{fontenc}
\usepackage[utf8]{inputenc}
\usepackage[main=english]{babel}
\usepackage{csquotes}
\usepackage{microtype}

\usepackage{mathtools}
\mathtoolsset{centercolon,mathic}
\allowdisplaybreaks[2] 
\usepackage{amssymb}
\usepackage{amsthm}
\usepackage{thmtools}
\usepackage{derivative}
\usepackage{enumitem}
\usepackage{booktabs}
\usepackage{makecell}

\newcommand{\statement}[1]{\paragraph{#1}\pdfbookmark[1]{#1}{#1}} 

\renewcommand{\Delta}{\varDelta}

\renewcommand{\Phi}{\varPhi}

\renewcommand{\Psi}{\varPsi}

\renewcommand{\Lambda}{\varLambda}

\renewcommand{\Gamma}{\varGamma}

\renewcommand{\Omega}{\varOmega}

\DeclarePairedDelimiter{\intervaloo}{\lparen}{\rparen}

\DeclarePairedDelimiter{\intervalco}{\lbrack}{\rparen}
\DeclarePairedDelimiter{\intervalcc}{\lbrack}{\rbrack}

\DeclarePairedDelimiter{\paren}{\lparen}{\rparen}
\DeclarePairedDelimiter{\abs}{\lvert}{\rvert}
\DeclarePairedDelimiter{\norm}{\lVert}{\rVert}

    \newcommand{\VERT}[1]{#1|\mkern-1.5mu#1|\mkern-1.5mu#1|}

    \ExplSyntaxOn
    \makeatletter

    \MT_delim_default_inner_wrappers:n{\tnorm}

    \NewDocumentCommand{\tnorm}{ s o m }{
        \IfBooleanTF{#1}{
        \MT_delim_tnorm_star_wrapper:nnn%
            {\VERT{\bgroup\left}}{#3}{\VERT{\aftergroup\egroup\right}}
        }{
            \IfValueTF{#2}{
                \@nameuse{MT_delim_tnorm_nostarscaled_wrapper:nnn}%
                    {\VERT{\@nameuse {\MH_cs_to_str:N #2 l}}}
                    {#3}
                    {\VERT{\@nameuse {\MH_cs_to_str:N #2 r}}}
            }{
                \MT_delim_tnorm_nostarnonscaled_wrapper:nnn%
                    {\VERT{}}
                    {#3}
                    {\VERT{}}
            }
        }
    }

    \makeatother
    \ExplSyntaxOff

\DeclarePairedDelimiterX{\innerproduct}[2]{\langle}{\rangle}{#1,#2}
\DeclarePairedDelimiterX{\innerp}[2]{\langle}{\rangle}{#1,#2}

\DeclarePairedDelimiter{\List}{\{}{\}}

\DeclarePairedDelimiter{\floor}{\lfloor}{\rfloor}

\DeclarePairedDelimiterXPP{\dd}[1]{d}{\lparen}{\rparen}{}{#1}
\DeclarePairedDelimiterXPP{\dist}[1]{d}{\lparen}{\rparen}{}{#1}
\DeclarePairedDelimiterXPP{\diam}[1]{\SYMdiam}{\lparen}{\rparen}{}{#1}
\DeclarePairedDelimiterXPP{\Exp}[1]{\exp}{\lparen}{\rparen}{}{#1}

\makeatletter
\let\bBigg@@\bBigg@
\renewcommand{\bBigg@}[2]{{%
  \mathchoice
    {\bBigg@@{#1}{#2}}%
    {\bBigg@@{#1}{#2}}%
    {\big@size=.5\big@size\bBigg@@{#1}{#2}}%
    {\big@size=.3\big@size\bBigg@@{#1}{#2}}}}%
\makeatother

\DeclareDocumentCommand{\trace}{s o e{_} m}{
    \IfValueTF{#3}
        {\Tr_{#3}}
        {\Tr}%
    \IfBooleanTF{#1}
        {\paren*{#4}}
        {
            \IfValueTF{#2}
                {\paren[#2]{#4}}
                {\paren{#4}}%
        }%
}
\DeclareDocumentCommand{\Trace}{s o e{_} m}{
    \IfValueTF{#3}
        {\Tr_{#3}}
        {\Tr}%
    \IfBooleanTF{#1}
        {\paren*{#4}}
        {
            \IfValueTF{#2}
                {\paren[#2]{#4}}
                {\paren{#4}}%
        }%
}

\providecommand\given{}

\newcommand\SetSymbol[1][]{%
    \nonscript\,#1\vert
    \allowbreak
    \nonscript\,
    \mathopen{}}
\DeclarePairedDelimiterX\Set[1]\{\}{%
    \renewcommand\given{%
        \SetSymbol[\delimsize]}
    \nonscript\,
    #1
    \nonscript\,
}

\ExplSyntaxOn
\makeatletter
\cs_generate_variant:Nn \tl_if_eq:nnT {onT}
\tl_new:N \__scale_new_delimsize_tl
\cs_new:Nn \__scale_make_bigger_tl:N {
    \tl_if_eq:onT {#1} {\@MHempty} { \tl_set:Nn \__scale_new_delimsize_tl { big } }
    \tl_if_eq:onT {#1} {}          { \tl_set:Nn \__scale_new_delimsize_tl { big } }
    \tl_if_eq:onT {#1} {\big}      { \tl_set:Nn \__scale_new_delimsize_tl { Big } }
    \tl_if_eq:onT {#1} {\Big}      { \tl_set:Nn \__scale_new_delimsize_tl { bigg } }
    \tl_if_eq:onT {#1} {\bigg}     { \tl_set:Nn \__scale_new_delimsize_tl { Bigg } }
    \tl_if_eq:onT {#1} {\middle}   { \tl_set:Nn \__scale_new_delimsize_tl { middle } }
}
\cs_new:Nn \scale_make_bigger_m:N {
    \__scale_make_bigger_tl:N #1
    \use:c { \tl_use:N \__scale_new_delimsize_tl }
}
\cs_new:Nn \scale_make_bigger_l:N {
    \__scale_make_bigger_tl:N #1
    \tl_if_eq:NnTF \__scale_new_delimsize_tl {middle} {
        \tl_set:Nn \__scale_new_delimsize_tl {left}
    }{
        \tl_put_right:Nn \__scale_new_delimsize_tl {l}
    }
    \use:c { \tl_use:N \__scale_new_delimsize_tl }
}
\cs_new:Nn \scale_make_bigger_r:N {
    \__scale_make_bigger_tl:N #1
    \tl_if_eq:NnTF \__scale_new_delimsize_tl {middle} {
        \tl_set:Nn \__scale_new_delimsize_tl {right}
    }{
        \tl_put_right:Nn \__scale_new_delimsize_tl {r}
    }
    \use:c { \tl_use:N \__scale_new_delimsize_tl }
}
\DeclarePairedDelimiterXPP{\pdd}[1]{\scale_make_bigger_l:N\delimsize\lparen d}{\lparen}{\rparen}{\scale_make_bigger_r:N\delimsize\rparen}{#1}
\DeclarePairedDelimiterXPP{\pdist}[1]{\scale_make_bigger_l:N\delimsize\lparen d}{\lparen}{\rparen}{\scale_make_bigger_r:N\delimsize\rparen}{#1}
\DeclarePairedDelimiterXPP{\pdiam}[1]{\scale_make_bigger_l:N\delimsize\lparen \SYMdiam}{\lparen}{\rparen}{\scale_make_bigger_r:N\delimsize\rparen}{#1}

\makeatother
\ExplSyntaxOff

\usepackage[svgnames]{xcolor}

\usepackage[
    colorlinks,
    allcolors=col_1,
    bookmarksnumbered,
    bookmarksopen,
    bookmarksopenlevel=2,
    bookmarksdepth=3,
    pdfusetitle=true,
    breaklinks=true,
]{hyperref}

\makeatletter
\@ifpackageloaded{hyperref}{
    \usepackage[capitalise,noabbrev]{cleveref}
    \usepackage{orcidlink}
}{
    \usepackage[capitalise,noabbrev,poorman]{cleveref}
    \newcommand{\texorpdfstring}[2]{#1}
    \newcommand{\href}[2]{#2}
    \newcommand{\hypersetup}[1]{}
    \newcommand{\orcidlink}[1]{ORCiD}
    \newcommand{\pdfbookmark}[1]{}
}
\makeatother

\crefdefaultlabelformat{#2#1#3}
\crefname{equation}{}{}

\usepackage{caption}
\usepackage{subcaption}

\newcommand{\sumstack}[2][]{\ifstrempty{#1}{\sum_{\substack{#2}}}{\smashoperator[#1]{\sum_{\substack{#2}}}}}

\newcommand{\e}{{\mathrm{e}}}
\newcommand{\I}{\mathrm{i}}
 
\newcommand{\C}{\mathbb{C}}
\newcommand{\N}{\mathbb{N}}
\newcommand{\Z}{\mathbb{Z}}
\newcommand{\HS}{{\mathcal{H}}}
\newcommand{\alg}{\mathcal{A}}
\newcommand{\algloc}{\alg_{\mathup{loc}}}
\newcommand{\unit}{\mathbf{1}}

\newcommand\cexpsym{\mathbb{E}}
\ExplSyntaxOn
\DeclareDocumentCommand{\cexp}{s o m m}{%
    \cexpsym\c_math_subscript_token{#3}
    \IfBlankF{#4}
    {
        \exp_last_unbraced:Ne \paren {\IfBooleanT{#1}{*}\IfValueT{#2}{[\exp_not:N #2]}} {#4}
    }
}
\ExplSyntaxOff

\newcommand{\calC}{\mathcal{C}}

\DeclareMathOperator{\Tr}{Tr}
\newcommand{\SYMdiam}{\operatorname{diam}}
\DeclareMathOperator{\Cov}{Cov}

\undef{\Re}
\DeclareMathOperator{\Re}{Re}

\newcommand{\cupdot}{\mathbin{\mathaccent\cdot\cup}}

\let\oldoverline\overline
\newcommand{\closure}[1]{\oldoverline{#1}}

\let\oldcirc\circ
\newcommand{\interior}[1]{{#1}^\oldcirc}
\def\circ{\oldcirc \PackageWarning{symbols}{Please use \protect\interior{...} instead of ...^{\protect\circ}}}

\newcommand{\suppv}{B_R}
\newcommand{\suppvshift}[1]{B_R(#1)}
\newcommand{\conn}{R}

\newcommand{\suchthat}{\mathpunct{\ordinarycolon}}

\newcommand{\quadtext}[1]{\quad\text{#1}\quad}
\newcommand{\qquadtext}[1]{\quad\quadtext{#1}\quad}

\newcommand{\Alignindent}{\hspace*{2em}&\hspace*{-2em}}

\newcommand{\mathup}[1]{\mathrm{#1}}

\newcommand{\fnfrac}[2]{\text{\footnotesize\(\displaystyle\frac{#1}{#2}\)}}

\let\oldsetminus\setminus
\newbox\mybox
\newcommand\cutsetminus[1]{%
    \setbox\mybox\hbox{\(#1\oldsetminus\)}%
    \ht\mybox=0pt%
    \usebox\mybox%
}
\renewcommand\setminus{%
    \mathbin{%
        \mathchoice%
            {\displaystyle\oldsetminus}
            {\textstyle\oldsetminus}
            {\cutsetminus{\scriptstyle}}
            {\cutsetminus{\scriptscriptstyle}}
    }%
}

\ExplSyntaxOn
\makeatletter
\appto{\thmt@newtheorem@postdefinition}{
    \cs_if_exist:cT {c@\thmt@envname}{%
        \exp_args:Nco \renewcommand {theH\thmt@envname} {\theHsection.\arabic{\thmt@envname}}
    }%
}
\makeatother
\ExplSyntaxOff

\theoremstyle{plain}

\declaretheorem[
    name=Theorem,
    numberwithin=section
]{theorem}

\declaretheorem[
    name=Lemma,
    sibling=theorem,
]{lemma}

\declaretheorem[
    name=Corollary,
    sibling=theorem,
]{corollary}

\declaretheorem[
    name={Main Result (informal)},
    numbered=no,
]{mainresult}

\declaretheorem[
    name={Problem},
    numbered=no,
]{mainproblem}

\theoremstyle{definition}

\declaretheorem[
    name=Definition,
    sibling=theorem,
]{definition}

\theoremstyle{remark}
\declaretheorem[
    name=Remark,
    sibling=theorem,
    qed=\(\diamond\),
]{remark}

\declaretheorem[
    name=Example,
    sibling=theorem,
    qed=\(\diamond\),
]{example}

\usepackage{xurl}

\usepackage{dsfont}

\definecolorset{HTML}{col_}{}{%
    1,0072B2;%
    2,D55E00;%
    3,E69F00
}

\usepackage[
    bibstyle=phys,
    citestyle=numeric,
    biblabel=brackets,
    mincitenames=1,
    maxcitenames=3,
    minalphanames=3,
    maxalphanames=4,
    minbibnames=3,
    maxbibnames=10,
    arxiv=abs,
    eprint=true,
    doi=false,
    url=false,
    sorting=nty,
    useprefix=true,
    alldates=year,
    urldate=long,
]{biblatex}
\renewbibmacro*{note+pages}{%
    \printfield{note}%
    \setunit{\bibpagespunct}%
    \printfield{pages}%
    \newunit%
}
\makeatletter
\DeclareFieldFormat{eprint:arxiv}{%
    \ifhyperref
    {\href{https://arxiv.org/\abx@arxivpath/#1}{%
            arXiv\addcolon\linebreak[2]#1}}
    {arXiv\addcolon
        \nolinkurl{#1}%
    }
}
\makeatother

\renewbibmacro*{maintitle+title}{%
    \iffieldsequal{maintitle}{title}{%
        \clearfield{maintitle}%
        \clearfield{mainsubtitle}%
        \clearfield{maintitleaddon}%
    }{%
        \iffieldundef{maintitle}{}{%
            \printtext[doi/url-link]{%
                \usebibmacro{maintitle}%
            }%
            \newunit\newblock
            \iffieldundef{volume}{}{%
                \printfield{volume}%
                \printfield{part}%
                \setunit{\addcolon\space}%
            }%
        }%
    }%
    \printtext[doi/url-link]{%
        \usebibmacro{title}%
    }%
    \newunit%
}

\DeclareBibliographyDriver{article}{%
    \usebibmacro{bibindex}%
    \usebibmacro{begentry}%
    \usebibmacro{author/translator+others}%
    \setunit{\labelnamepunct}\newblock
    \usebibmacro{title}%
    \newunit
    \printlist{language}%
    \newunit\newblock
    \usebibmacro{byauthor}%
    \newunit\newblock
    \usebibmacro{bytranslator+others}%
    \newunit\newblock
    \printfield{version}%
    \newunit\newblock
    \printtext[doi/url-link]{%
        \usebibmacro{journal+issuetitle}%
        \newunit
        \usebibmacro{byeditor+others}%
        \newunit
        \usebibmacro{note+pages}%
        \newunit\newblock
        \iftoggle{bbx:isbn}
        {\printfield{issn}}
        {}%
    }%
    \newunit\newblock
    \printfield{pubstate}%
    \setunit{\addspace}%
    \printfield{year}%
    \newunit\newblock
    \usebibmacro{doi+eprint+url}%
    \newunit\newblock
    \printfield{addendum}%
    \setunit{\bibpagerefpunct}\newblock
    \usebibmacro{pageref}%
    \newunit\newblock
    \iftoggle{bbx:related}
        {\usebibmacro{related:init}%
        \usebibmacro{related}}
        {}%
    \usebibmacro{finentry}%
}

\newcommand{\powerset}{\mathcal{P}_0}
\newcommand{\Cint}{C_{\mathup{int}}}

\newcommand{\Ht}{\tilde{H}}

\newcommand{\Zt}{\tilde{Z}}
\newcommand{\ft}{\tilde{f}}
\newcommand{\gt}{\tilde{g}}
\newcommand{\Tt}{\tilde{T}}

\usepackage{url}

\makeatletter
\def\blindfootnote{\gdef\@thefnmark{}\@footnotetext}
\makeatother
\newcommand{\emaillink}[1]{\href{mailto:#1}{#1}}

\title{Uniform-in-temperature locality estimates for weakly interacting quantum systems}

\author{
    Arka Adhikari%
    \texorpdfstring{%
        \,\orcidlink{0000-0003-4260-0015}
        \footnote{
            \parbox[t]{.75\textwidth}{
                Department of Mathematics, University of Maryland-College Park,
                \\
                College Park, MD, USA.
            }
        }
    }{}%
    \and 
    Joscha Henheik%
    \texorpdfstring{%
        \,\orcidlink{0000-0003-1106-327X}
        \footnote{
            \parbox[t]{.75\textwidth}{
                University of Geneva,
                \\
                Rue du Conseil Général 7-9, 1205 Geneva,
                Switzerland.
            }
        }
    }{}%
    \and
    Marius Lemm%
    \texorpdfstring{%
        \,\orcidlink{0000-0001-6459-8046}
        \footnote{
            \parbox[t]{.75\textwidth}{
                Department of Mathematics,
                University of Tübingen,
                \\
                Auf der Morgenstelle 10,
                72076 Tübingen,
                Germany. 
            }
        }
    }{}%
    \and 
    Tom Wessel%
    \texorpdfstring{%
        \,\orcidlink{0000-0001-7593-0913}
        \footnotemark[3]
    }{}%
}

\date{January 21, 2026}

\begin{document}

\bgroup
\hypersetup{hidelinks}
\maketitle\thispagestyle{empty}
\blindfootnote{
    Email:\quad%
    \parbox[t]{.7\textwidth}{
        \hypersetup{hidelinks}
        \emaillink{arkaa@umd.edu},
        \emaillink{joscha.henheik@unige.ch},
        \newline
        \emaillink{marius.lemm@uni-tuebingen.de},
        \emaillink{tom.wessel@uni-tuebingen.de}
    }
}
\egroup

\begin{abstract}
    The locality of thermal quantum states has emerged as a key input for applications to thermalization, response theory, and efficient simulability.
    Locality is either captured by the decay of correlations or by local indistinguishability, which allows to approximate local expectation values by those of local thermal states.
    Most techniques for deriving locality bounds deteriorate at small temperature, a physically highly relevant regime and so it is of interest to identify conditions for uniform-in-temperature bounds.
    Here we prove that a class of weakly interacting quantum Hamiltonians satisfies exponential decay of correlations and local indistinguishability \emph{uniformly in the temperature}.
    The proof uses a low-temperature cluster expansion and a quantum version of a probabilistic swapping trick developed by the first author and Cao~\cite{AC2024} in the context of lattice gauge theories.
\end{abstract}

\section{Introduction}
A central problem in the study of equilibrium quantum many-body physics is to determine the locality properties of thermal (Gibbs) states
\begin{equation*}
    \rho^\beta = \frac{\e^{-\beta H} }{ \trace{\e^{-\beta H}}}
\end{equation*}
where \(H\) is the Hamiltonian operator, assumed to be a sum of local interaction terms, and \(\beta=1/T\) denotes the inverse temperature.
The traditional way to express locality is through exponential decay of correlations (DoC), also known as \enquote{clustering} of correlations.
DoC is the statement that two bounded local observables \(A\) and \(B\) supported on distinct spatial regions \(X\) and \(Y\), respectively, satisfy
\begin{equation}
    \label{eq:introdoc}
    \abs[\big]{
        \trace{A \, B \, \rho^\beta}
        - \trace{A \, \rho^\beta} \, \trace{B \, \rho^\beta}
    }
    \lesssim
    \norm{A} \, \norm{B} \, \e^{- d(X,Y)/\xi(\beta)}
\end{equation}
for a suitable temperature-dependent decay rate \(\xi(\beta)>0\), whose sharp value is called the inverse correlation length.
As in classical statistical mechanics, clustering of correlations is intimately connected to the absence of phase transitions, a longstanding and notoriously difficult topic in mathematical physics (see e.g.~\cite{ginibre1969existence,robinson1969proof,dyson1978phase}).

A new perspective has emerged in the past 15~years originating from quantum information theory.
From this new perspective, the locality of thermal states is essential
to break up a quantum many-body system into local pieces, which can be separately prepared, simulated, manipulated, or measured -- depending on the precise practical task at hand.
Thus, locality provides a powerful tool to address the fundamental curse of dimensionality that arises in any quantum many-body problem.
Celebrated examples of such uses of locality include the efficient classical simulability of thermal states~\cite{molnar2015approximating,harrow2020classical,alhambra2021locally,cirac2021matrix,bravyi2022quantum,fawzi2023subpolynomial}, their efficient tomography~\cite{rouze2024efficient} and sampling~\cite{temme2011quantum,ding2025efficient}, as well as the rapid equilibration of open quantum systems~\cite{kuwahara2020eigenstate,bardet2023rapid,,bardet2024entropy,kochanowski2025rapid}, which is relevant for exploring possible quantum memories and for the efficient preparation of thermal states on quantum devices~\cite{BK2019,kato2019quantum,chen2023quantum,rouze2025efficient}.
A strong form of clustering was shown to imply the area law for the entanglement entropy~\cite{BH2015}.
Related locality bounds have also been used in condensed matter physics to prove quantization of the Hall conductance~\cite{hastings2015quantization, giuliani2017universality,bachmann2020many, kapustin2020hall, WMM+2025}.

It turns out that for many of these applications one does not directly require decay of correlations, but instead a different locality property of the thermal state which is called \emph{local indistinguishability}~(LI)~\cite{KGK2014,BK2019,bravyi2022quantum,BCP2022,CMTW2023}, which will play a prominent role in this paper.
LI says that when calculating the expectation value of a local observable, one can replace the global thermal state by a local thermal state.
More precisely, given nested spatial regions \(Y\subset \Lambda'\subset\Lambda\) and a bounded local observable~\(B\) supported on region~\(Y\), LI is a bound of the difference between expectation values in the thermal states \(\rho^\beta_\Lambda\) and \(\rho^\beta_{\Lambda'}\) on the different regions, i.e.
\begin{equation}
    \label{eq:introLI}
    \abs[\big]{
        \trace{ B \, \rho^\beta_\Lambda}
        - \trace{B\, \rho^\beta_{\Lambda'}}
    }
    \lesssim
    \norm{B} \, \e^{- d(Y,\Lambda'\setminus\Lambda)/\xi_{\mathrm{LI}}(\beta)}
    .
\end{equation}
Here, \(\xi_{\mathrm{LI}}(\beta)\) is another suitable temperature-dependent decay rate, which may or may not be the same as \(\xi(\beta)\) in~\eqref{eq:introdoc}.
LI was originally introduced by Kliesch et al.\ as establishing the \emph{locality of temperature}~\cite{KGK2014}, because it indeed allows measuring the system's global temperature based on local information only.
LI also has many applications to address the curse of dimensionality, especially for simulating thermal states, either on a classical or on a quantum device~\cite{BK2019} and for efficiently preparing thermal states~\cite[Section~III.B]{chen2025quantum} as we discuss further in Section~\ref{sect:applications}.
For further background on LI and its uses, we refer to the reviews~\cite{KR2019,alhambra2023quantum}.

\bigskip

Since the locality properties of thermal states are of such fundamental importance for many applications, their mathematical derivation has developed into a large subfield in its own right.
Since locality is intimately connected to the absence of phase transitions, most existing results treat the case of zero temperature for gapped ground states~\cite{HK2006,NS2006,RS2019,WH2022}, 1D systems~\cite{Araki1969,BCP2022,PP2023,kuwahara2024clustering}, or the high temperature regime by cluster expansions~\cite{park1982cluster,KGK2014, FU2015, KR2019,kuwahara2020clustering,kato2025clustering}.

The low-temperature regime is comparatively less studied, even though it is highly relevant for condensed matter physics and modern quantum simulation platforms in ultracold gases~\cite{bloch2008many}.
Of course, locality at low temperatures in dimension greater than one is a more subtle topic because phase transitions can occur.
Proving DoC at low temperature for suitable systems is a longstanding topic in mathematical physics with many significant contributions, both for quantum spin systems~\cite{TY1983,messager1996low,borgs1996low,DFF1996,DFF1996a,frohlich2001quantum} and fermions~\cite{Hastings2004decay,HK2006,RS2019,frohlich2001quantum,GLMP2024}.
However, proving local indistinguishability (LI) uniformly at low temperature does not seem to have been considered so far either in the mathematical physics or quantum information theory literature.

This raises the following problem.

\begin{mainproblem}
    \label{statement:mainproblem}
    Identify conditions under which DoC and LI hold uniformly in temperature for quantum systems in any spatial dimension.
\end{mainproblem}

In particular, the challenge is to obtain bounds on the \emph{correlation lengths}~\(\xi\) in~\eqref{eq:introdoc} and~\(\xi_{\mathrm{LI}}\) in~\eqref{eq:introLI} that are uniform in~\(\beta\) and thus bounded in the low-temperature~\(\beta\to\infty\) regime.
By analogy with classical statistical mechanics, it is clear that the relevant condition should exclude phase transitions at low temperature and we discuss this point in detail in Section~\ref{sec:discussion}.

Our goal in this paper is to tackle this problem head-on by presenting conceptually simple and robust proofs of both DoC and LI with a decay rate that is independent of~\(\beta\) and so, in particular, uniform as~\(\beta\to\infty\).
To achieve this, we develop a new analytical argument that combines a low-temperature cluster expansion with a quantum version of a swapping trick originating in probability theory~\cite{AC2025}; see Section~\ref{subsec:swapping} for a brief sketch of the main idea.

\begin{remark}[Relation between DoC and LI]
    \label{rmk:DoC_vs_LI}
    Let us comment on the relation between DoC and LI and why one may consider LI as the stronger property.
    First, a general form of LI implies DoC with the same constants because the local truncation of the thermal states decorrelates the observables; see \cref{rem:LI-implies-DoC} for more details.
    Conversely, it has been shown~\cite{BK2019,CMTW2023} that at every \emph{fixed} temperature~\(T\) %
    DoC uniformly%
    \footnote{%
        This kind of \enquote{uniform} DoC is standard and sometimes referred to as \enquote{uniform clustering}, where uniformity refers to the choice of \(\Lambda'\subset \Lambda\).
        For the systems we consider, we indeed prove results uniform in~\(\Lambda\).
    }
    in \(\Lambda'\) implies LI\@.
    Unfortunately, by this route, the constants in the resulting LI statement depend adversely on~\(T\).
    More precisely, uniform-in-temperature DoC implies LI with a prefactor and correlation length that scale polynomially in \(\beta\) as~\(\beta\to \infty\), which poses a problem for the applications we described before.
    The underlying reason for this effect is that the main tool of~\cite{BK2019,CMTW2023} is quantum belief propagation~\cite{hastings2007quantum}, a differential equation to describe deformed Gibbs states, and this is well-known to produce constants that diverge as~\(T\to0\).
    To summarize, uniform-in-temperature LI implies uniform-in-temperature DoC, while, conversely, uniform-in-temperature DoC implies LI with constants that diverge as~\(T\to 0\).
\end{remark}

\subsection{Summary of main results}
In this paper, we prove exponential DoC and LI at all inverse temperatures \(\beta \in \intervalco{0, \infty}\) for lattice systems on \(\Lambda \Subset \Z^D\) with finite-dimensional on-site Hilbert spaces.
The considered Hamiltonians are of the form
\begin{equation*}
    H
    =
    H^0
    + V
    =
    \sum_{x\in \Lambda} h_x
    + \sumstack[lr]{x \in \Lambda : \\ B_R(x) \subset \Lambda} v_x
    .
\end{equation*}
Here, every \(h_x \ge 0\) acts only on the site \(x \in \Lambda\) and has a unique ground state separated from the rest of the spectrum by a gap of size at least one, and every \(v_x\) is a small finite-range interaction supported on \(B_R(x)\), a ball of radius \(R\) around site \(x\).
In particular, we do not assume any sort of periodicity or translation invariance.

Additionally, we assume that every \(v_x\) is relatively form bounded with respect to the sum of all \(h_x\) in \(B_R(x)\) in the following sense: There exists a small \(a \in (0,1)\) such that
\begin{equation}
    \label{eq:formbound}
    \abs{\innerp{\psi}{v_x \, \psi}}
    \leq
    a \, \innerp[\Big]{\psi}{\sumstack[lr]{z\in \suppvshift{x}} h_z \, \psi}
    \qquad\text{for all states \(\psi\)}
    .
\end{equation}
We explain the motivation for studying these Hamiltonians in Section~\ref{sec:discussion}.

Roughly speaking, the form bound~\eqref{eq:formbound} ensures that the spectrum of~\(H\) is contained in~\(\intervalco{0, \infty}\) (as is the case for~\(H^0\)), the ground state of \(H\) remains a gapped product state, and all the (non-negative) eigenvalues~\(\{E_j^0\}\) of~\(H^0\) remain \enquote{of the same order} when perturbed by~\(V\).
In fact, a simple application of the min-max-principle (or alternatively of~\cite[Theorem~3.6 in Chapter VII]{Kato1995}) shows that the \(j^{\mathrm{th}}\)~eigenvalue of~\(H\) lies in the interval \(\intervalcc{(1-a) \, E_j^0, (1+a) \, E_j^0}\).
This ensures, in particular, that the many-body density of states (mbDoS) of~\(H\) behaves similarly to the mbDoS of the non-interacting Hamiltonian~\(H^0\): The low- and high-energy states have a small mbDoS, while for intermediate energies the mbDoS is large.

As our main results, we obtain the following bounds (see Theorems~\ref{thm:DoC} and~\ref{thm:LI}).

\begin{mainresult}
    For small enough \(a>0\) as in~\eqref{eq:formbound}, there exist constants \(C_1\),~\(C_2\)~\(\xi\),~\(\xi_{\mathrm{LI}} > 0\), such that the system satisfies decay of correlations~(DoC)
    \begin{equation*}
        \abs{
            \trace{A \, B \, \rho^\beta_\Lambda}
            - \trace{A \, \rho^\beta_\Lambda} \, \trace{B \, \rho^\beta_\Lambda}
        }
        \le
        C_1 \, \e^{C_2 \, (\abs{X}+\abs{Y})} \, \norm{A} \, \norm{B} \, \e^{- \dist{X,Y}/\xi}
    \end{equation*}
    and local indistinguishability~(LI)
    \begin{equation*}
        \abs{
            \trace{B \, \rho^\beta_\Lambda}
            - \trace{B \, \rho^\beta_{\Lambda'}}
        }
        \le
        C_1 \, \e^{C_2 \, \abs{Y}} \, \norm{B} \, \e^{- \dist{Y,\Lambda\setminus\Lambda'}/\xi_{\mathrm{LI}}}
    \end{equation*}
    for all \(\Lambda'\subset \Lambda \Subset \Z^D\), \(\beta>0\) and observables \(A\) and \(B\) supported on \(X\) and \(Y\), respectively.
\end{mainresult}

Notice that \(\xi\) and \(\xi_{\mathrm{LI}}\) are independent of~\(\beta\).
We remark that there is a third notion of locality of Gibbs states known as the \emph{local perturbations perturb locally} (LPPL) principle, which asserts that local perturbations of the Hamiltonian affect the associated thermal state only locally~\cite{bachmann2012automorphic,RS2015, henheik2022local,bachmann2022stability,CMTW2025}.
Our method adapts to LPPL and we discuss this further in Section~\ref{sec:LPPL}.

\subsection{Discussion}
\label{sec:discussion}
As in classical statistical mechanics, if exponential DoC for \(\rho^\beta\) breaks down at a certain critical \(\beta_{\mathup{c}}\), then this indicates a phase transition occurring at this inverse temperature.
This explains why proving DoC and LI for Gibbs states is necessarily a temperature-dependent and thus somewhat delicate task -- especially in higher dimensions.
Therefore, to prove DoC and LI, it is necessary to identify a condition which excludes a phase transition in the temperature regime under consideration.
There are different approaches to excluding phase transitions.

\begin{itemize}
    \item
        In arbitrary dimensions, results are usually about sufficiently \emph{high temperature}, i.e.~for \(\beta\) below a certain (universal) critical inverse temperature \(\beta_*\), no phase transitions are possible and DoC can be derived by cluster expansion techniques~\cite{park1982cluster,KGK2014, FU2015, KR2019}.

    \item A special case arises in one spatial dimension, where it is folklore wisdom that phase transitions can only occur at zero temperature.
        Indeed, Araki~\cite{Araki1969} showed in pioneering work in 1969 that exponential DoC holds at any positive temperature in one-dimensional, finite-range, translation-invariant systems; see~\cite{BCP2022,PP2023,KK2024} for recent extensions.
        While the one-dimensional case covers any positive temperature, the possibility of a phase transitions at zero temperature affects the \(T\to 0\) (or \(\beta\to\infty\)) behaviour of the bound.
        This is expressed through the correlation length~\(\xi(\beta)\) which features as the inverse decay rate in~\eqref{eq:introdoc}.
        The best known upper bound \(\xi(\beta) \lesssim \Exp{c \, \beta}\) was proved recently in~\cite{KK2024} and divergence as \(\beta\to\infty\) can indeed occur.
        For example, in the classical one-dimensional Ising model with coupling strength \(J > 0\), the correlation length of the two-point spin-spin covariance can explicitly be computed to be given by \(\xi(\beta) = -1/\log(\tanh(\beta J)) \sim \tfrac{1}{2} \, \e^{2 J \beta}\) as \(\beta \to \infty\).
        This shows that deriving a uniform-in-temperature bound on the correlation length is subtle even in one dimension.

\end{itemize}

We aim for a uniform bound on the correlation length in any dimension and so these two standard approaches to exclude phase transitions are \emph{unavailable to us}.

One thus needs a different idea to exclude phase transitions and this is to suitably \enquote{perturb} a stable classical phase with a non-commuting perturbation.
A stable classical phase can be implemented for example by breaking a symmetry.
Indeed, taking another hint from the classical one-dimensional Ising model, notice that its behaviour drastically changes when breaking the symmetry through a constant external field of strength~\(h > 0\).
This breaks the two-fold ground state degeneracy present for \(h=0\) and leads to a correlation length that is uniformly bounded in the low temperature regime.
The idea of adding quantum perturbations to classical systems and deriving DoC was first implemented for special models~\cite{TY1983,messager1996low}.
A general approach was then developed around the same time by Borgs, Koteck\'y and Ueltschi~\cite{borgs1996low} and Datta, Fr\"ohlich and Fernandez~\cite{DFF1996,DFF1996a} for translation-invariant systems through quantum versions of the machinery of Pirogov-Sinai theory.
These works obtain DoC
in dimension at least two for translation-invariant quantum perturbations of classical Hamiltonians with finitely many periodic gapped ground states.
Intermediate temperature ranges have been considered as well~\cite{frohlich2001quantum}, also by a modified Pirogov-Sinai theory.

Our model is also a quantum perturbation of a classical Hamiltonian as in~\cite{borgs1996low,DFF1996,DFF1996a}.
However, there are a few important differences: We consider the case that the classical Hamiltonian has a unique ground state, while these works consider classical Hamiltonian with multiple phases which are then distinguished by suitable boundary conditions.
Importantly, our method is completely different because it tackles DoC and LI directly through a swapping trick (see Section~\ref{subsec:swapping}), while these works develop an involved quantum version of Pirogov-Sinai theory, whose convergence then implies DoC in their setting.
Our new approach via the swapping trick leads to a fully self-contained and, we believe, conceptually simple proof of DoC and the first proof of LI uniformly in temperature.

Another advantage of the fact that our new method is very direct is that the proof is robust in new ways.
First, our proof also works in \emph{one spatial dimension}.
This means our result also improves the correlation length in recent DoC results in the low-temperature regime~\cite{BCP2022,PP2023,Kuwahara2024} for our class of weakly interacting quantum Hamiltonians, and we obtain the first uniform bound on the correlation length~\(\xi(\beta)\).
Second, we do not assume that the classical Hamiltonian and the quantum perturbation are \emph{translation-invariant}, which was needed for the Pirogov-Sinai theory~\cite{borgs1996low,DFF1996}.
This robustness arises because the swapping trick exploits local cancellations that are adapted to the local structure of the system.
In particular, we can treat for the first time \emph{disordered} Hamiltonians, like the disordered XXZ chain that we describe now.

\begin{example}
    \label{ex:XXZ}
    A paradigmatic example Hamiltonian satisfying our conditions is the XXZ model with (random) external field on some \(\Lambda \Subset \Z^D\).
    More precisely, let \(\sigma^i_x\) be the \(i\)-th Pauli matrix acting only on site \(x \in \Lambda\), and define the \emph{ladder operators} and the \emph{number operator} acting on \(x \in \Lambda\) as \(\sigma^\pm_x = \frac{1}{2}(\sigma^1_x \pm \I \sigma^2_x)\) and \(\mathcal{N}_x = \frac{1}{2}(\unit_x - \sigma^3_x)\), respectively.
    Then the XXZ~Hamiltonian with (random) external field acting on the Hilbert space \(\bigotimes_{x \in \Lambda} \C^2\) is given by
    \begin{equation*}
        H_{\mathrm{XXZ}} = H^0 + V
    \end{equation*}
    with
    \begin{equation}
        \label{eq:XXZ}
        \begin{gathered}
            H^0
            :=
            \sum_{x \in \Lambda} (1 + \lambda \, \omega_x) \, \mathcal{N}_x
            \mathrlap{\qquad\text{and}}\\
            V
            :=
            \sum_{x,y \in \Lambda} J_{12}(x,y) \, (\sigma^+_x \sigma_y^- + \sigma_x^- \sigma_y^+)
            + \sum_{x,y \in \Lambda} J_3(x,y) \, \mathcal{N}_x \, \mathcal{N}_y
            ,
        \end{gathered}
    \end{equation}
    where \(\List{\omega_x}_{x \in \Lambda}\) is a collection of random variables supported on \([0,1]\) and the parameter \(\lambda \ge 0\) modulates the strength of the randomness.
    The self-adjoint coupling matrices \(J_{12}\) and \(J_3\), i.e.~satisfying \(J_{12}(y,x) = \overline{J_{12}(x,y)}\) and \(J_{3}(y,x) = \overline{J_{3}(x,y)}\), are assumed to be of \emph{finite range}.
    That is, there exists \(R > 0\) such that \(J_{12}(x,y) = J_3(x,y) = 0\) whenever \(\dist{x,y} \ge R\).

    By construction, every summand of~\(H^0\) in~\eqref{eq:XXZ} satisfies \((1 + \lambda \, \omega_x) \, \mathcal{N}_x \ge 0\) with ground state eigenvalue~\(0\) and excited eigenvalue~\(1 + \lambda \, \omega_x \ge 1\).
    Moreover, if \(\sup_{x,y \in \Lambda} \abs{J_{12}(x,y)}\) and \(\sup_{x,y \in \Lambda} \abs{J_{3}(x,y)}\) are small enough, we find a relative form bound as in~\eqref{eq:formbound} for the interaction \(V\) in~\eqref{eq:XXZ} in terms of \(H^0\).
    Hence, the XXZ~Hamiltonian \(H_{\mathrm{XXZ}}\) as in~\eqref{eq:XXZ} satisfies all the assumptions in our main result, and hence satisfies DoC and LI.
\end{example}

\begin{remark}
    At zero temperature, the existence of a spectral gap above the ground state implies DoC~\cite{HK2006,NS2006,RS2019,WH2022}.
    Naively, one may therefore think that a suitable condition to tackle the problem at low temperature is that there exists a unique gapped ground state.
    However, there are classical examples (an Ising chain perturbed by a single on-site field) which show that a spectral gap is insufficient to even tackle small positive temperature.
\end{remark}

\subsection{Applications and outlook}
\label{sect:applications}
Our result on uniform LI, Theorem~\ref{thm:LI}, implies that the \enquote{locality of temperature} found in~\cite{KGK2014} (i.e., the fact that the temperature of the global thermal state equals the temperature of its local approximation) is in fact a uniform property that does not deteriorate at small temperature.
Second, our uniform LI impacts the efficient preparation of thermal states, as can be seen by replacing the use of Theorem~III.2 with our Theorem~\ref{thm:LI} in the proof of Corollary~III.3 of~\cite{chen2025quantum}.
Third, LI also improves classical simulability results for thermal states.
More precisely, let us focus on the problem of efficiently simulating local observables (e.g., Problem~3 in~\cite{bravyi2022quantum}), which also allows for efficient simulation of free energies (see Problem~1 and Lemma~12 in~\cite{bravyi2022quantum}).
Here, efficient simulation refers to algorithms that produce errors that are polynomially small in the number of sites \(\abs{\Lambda}\) in runtimes that are also (only) polynomially small in \(\abs{\Lambda}\).
One can use LI to obtain such an efficient algorithm for expectation values of local observables by proceeding as follows: One first applies LI with the distance parameter \(d(Y,\Lambda'\setminus \Lambda)\) logarithmic in the total system size \(\abs{\Lambda}\), which yields a polynomial error in \(\abs{\Lambda}\).
This step can now be performed uniformly in temperature for our system.
Afterwards, one runs a classical algorithm on the subsystem whose runtime versus error scales with the subsystem dimension which (on account of the logarithm) will be polynomial in \(\abs{\Lambda}\).
It is an interesting question in this context if the second, classical step can also be performed uniformly in temperature by further exploiting the cluster expansion method we devise here, but we leave its investigation to future work.

Apart from the various applications of LI to quantum information science partly discussed above, our original motivation to derive LI uniformly at low temperature was that we have identified it as a key property for developing a robust response theory for low-temperature interacting quantum systems, a topic whose mathematical investigation has only begun recently~\cite{greenblatt2024adiabatic,jakvsic2024note}.
To keep this paper focused on the general properties of DoC and LI, we will present these consequences of our Theorem~\ref{thm:LI} for response theory elsewhere~\cite{linrep}.

\section{Setup and main results}
After introducing the necessary mathematical framework in Section~\ref{sec:setup}, we present our main results on decay of correlations (Theorem~\ref{thm:DoC}) and local indistinguishability (Theorem~\ref{thm:LI}) in Section~\ref{sec:mainres}.

\subsection{Mathematical setup}
\label{sec:setup}

Consider the regular lattice~\(\Z^D\), for fixed~\(D \in \N\), equipped with the \(\ell^1\)-metric \(d\colon \Z^D \times \Z^D \rightarrow \N_0\).
We denote arbitrary subsets as~\(\Lambda\subset \Z^D\) (including equality) and finite subsets by~\(\Lambda' \Subset \Lambda\) (again including equality if~\(\Lambda\) is finite).
The set of all finite subsets is denoted \(\powerset(\Lambda)\).
The cardinality of a set \(\Lambda\Subset \Z^D\) is denoted by~\(\abs{\Lambda}\).
Given any two subsets \(X\), \(Y \subset \Z^D\) we denote by \(\dist{X,Y}\) their distance with respect to the metric~\(d\).
Likewise, we denote by \(\diam{X} := \sup_{x, y \in X} d(x,y)\) the diameter of~\(X\).

With every site \(x \in \Z^D\) we associate a finite-dimensional local Hilbert space \(\HS_x := \C^q\), with \(q\geq 2\), and the corresponding space of linear operators \(\alg_x:=\mathcal{B} \paren{\HS_x}\).
For each \(\Lambda \Subset \Z^D\) we define the Hilbert space \(\HS_{\Lambda}:=\bigotimes_{x \in \Lambda} \HS_x\), and denote the algebra of bounded linear operators on~\(\HS_{\Lambda}\) by \(\alg_{\Lambda} := \mathcal{B}(\HS_{\Lambda})\).
Due to the tensor product structure, we have \(\alg_{\Lambda}=\bigotimes_{x \in \Lambda} \mathcal{B} (\HS_x)\).
Hence, for \(X \subset \Lambda \Subset \Z^D\), any \(A \in \alg_{X}\) can be viewed as an element of~\(\alg_{\Lambda}\) by identifying~\(A\) with \(A \otimes \unit_{\Lambda \backslash X} \in \alg_{\Lambda}\), where~\(\unit_{\Lambda \backslash X}\) denotes the identity in~\(\alg_{\Lambda \backslash X}\).
Using this identification, we define the algebra of local operators
\begin{equation*}
    \algloc
    :=
    \bigcup_{\Lambda\Subset\Z^D} \alg_{\Lambda}
    \quad \text{and its completion} \quad
    \alg
    :=
    \overline{\algloc}^{\norm{\cdot}}
    .
\end{equation*}

We consider Hamiltonians composed of two parts \(H_\Lambda := H^0_\Lambda + V_\Lambda\).
The first is a sum of on-site terms \(h_x\in \alg_{\List{x}}\), each with unique gapped ground state \(\Omega_x\in \HS_x\) with gap at least \(1\), i.e.~we have
\begin{equation*}
    \innerp{\Omega_x}{h_x \, \Omega_x}=0
    \qquadtext{and}
    \inf_{\substack{\psi\in \HS_x\\\innerp{\psi}{\Omega_x}=0}}
    \frac{\innerp{\psi}{h_x \, \psi}}{\innerp{\psi}{\psi}}
    \geq
    1
    .
\end{equation*}
In particular,
\begin{equation*}
    H^0_\Lambda
    :=
    \sum_{x\in \Lambda} h_x
\end{equation*}
has ground state \(\Omega_\Lambda = \bigotimes_{x\in \Lambda} \Omega_x\) with gap \(1\).

For the second part of the Hamiltonian, let
\begin{equation*}
    B_r(x)
    :=
    \Set[\big]{
        y\in \Z^D \given \dist{x,y} \leq r
    }
\end{equation*}
denote the ball in \(\Z^D\).
We also abbreviate \(B_r = B_r(0)\) which will mainly be used for the size \(\abs{B_r} = \abs{B_r(x)} \leq (2r+1)^D\).
Then, fix a range \(R\in \N\) and define
\begin{equation*}
    V_\Lambda
    :=
    \sumstack[lr]{x\in \Lambda\suchthat\\\suppvshift{x}\subset \Lambda} \, v_x
\end{equation*}
as a sum of local terms \(v_x \in \alg_{\suppvshift{x}}\).
We denote \(\norm{v}_\infty = \sup_{x\in \Z^D} \norm{v_x}\).
Moreover, we assume that there exists \(a<1\) such that each \(v_x\) is relatively form bounded w.r.t.\ \(H^0_{\suppvshift{x}}\) in the sense that
\begin{equation}
    \label{eq:relative-boundedness-assumption}
    \abs{\innerp{\psi}{v_x \, \psi}}
    \leq
    \frac{a}{\abs{\suppv}} \, \innerp{\psi}{H^0_{\suppvshift{x}} \, \psi}
    \qquad\text{for all \(\psi\in \HS_{\suppvshift{x}}\).}
\end{equation}
A class of examples of Hamiltonians satisfying our assumption are Heisenberg XXZ Hamiltonians in a sufficiently strong external field in any dimension (see Example~\ref{ex:XXZ} where the disordered case was presented).
Then, it immediately follows that \(V_\Lambda\) is relatively form-bounded w.r.t.\ \(H^0_\Lambda\), namely
\begin{equation*}
    \abs{\innerp{\psi}{V_\Lambda \, \psi}}
    \leq
    \frac{a}{\abs{\suppv}} \, \sumstack[lr]{x\in \Lambda\suchthat\\\suppvshift{x}\subset \Lambda} \, \innerp{\psi}{H^0_{\suppvshift{x}} \, \psi}
    \leq
    a \, \innerp{\psi}{H^0_\Lambda \, \psi}
    \qquad\text{for all \(\psi\in \HS_\Lambda\).}
\end{equation*}
We point out that, at the cost of adjusting the constant \(a\) in the form bound, we can assume w.l.o.g.~that \(v_x \le 0\) for all \(x \in \Lambda\).
In fact, denoting \(\tilde{a} := a /\abs{B_R}\), we can write
\begin{equation*}
    H^0_{B_R(x)} + v_x
    =
    (1+\tilde{a}) \, H^0_{B_R(x)} + (-\tilde{a} \, H^0_{B_R(x)} + v_x)
    ,
\end{equation*}
where now \((-\tilde{a} \, H^0_{B_R(x)} + v_x) \le 0\).
Moreover, \((1+\tilde{a}) \, h_x\) satisfies the same conditions as \(h_x\) and \((-a \, H^0_{B_R(x)} + v_x)\) satisfies the same conditions as \(v_x\), although the form bound is now satisfied with \(2 \, \tilde{a}/(1+\tilde{a})\) instead of \(a\).
We will henceforth assume that \(v_x \le 0\) without further mentioning this.

Finally, we denote the partition functions
\begin{equation}
    \label{eq:partition-function}
    Z^0_\Lambda
    =
    \trace[\big]{\e^{-\beta H^0_\Lambda}}
    \qquadtext{and}
    Z_\Lambda
    =
    \trace[\big]{\e^{-\beta H_\Lambda}}
\end{equation}
and the Gibbs states
\begin{equation}
    \label{eq:Gibbs-state}
    \rho^0_\Lambda
    =
    \frac{\e^{-\beta H^0_\Lambda}}{Z^0_\Lambda}
    \qquadtext{and}
    \rho_\Lambda
    =
    \frac{\e^{-\beta H_\Lambda}}{Z_\Lambda}
    ,
\end{equation}
of the unperturbed and the full system, respectively.
Here and in the following, unless specified differently, every trace (e.g., the ones used in the definition of the partition functions in~\eqref{eq:partition-function}) is understood to be taken over~\(\alg_\Lambda\).

To state the results, we also need to introduce the notion of \(\conn\)-connected sets.
They naturally appear later as part of the cluster decomposition, see \cref{def:connected-to-each-other,def:connected-sets}.

\begin{definition}[\(\conn\)-connected sets]
    \label{def:connected-sets-setup}
    A subset \(X\subset \Lambda\) is an \emph{\(\conn\)-connected set}, if for every two points \(x\), \(y\in X\) there exists a sequence of points \(z_1=x\), \(z_2,\dotsc,z_{m}\in X\), \(z_{m+1}=y\) such that \(\dist{z_i,z_{i+1}} \leq 2\,R\) for all \(i\in \List{1,\dotsc,m}\).
\end{definition}

\subsection{Main results}
\label{sec:mainres}

We are now ready to formally state our main results, and begin with decay of correlations (DoC) uniformly in the temperature.
Recall that \(D\) denotes the spatial dimension and \(q\) denotes the dimension of the local Hilbert space.

\begin{theorem}[Decay of correlations]
    \label{thm:DoC}
    Let \(D\), \(q\), \(R\in \N\) and \(\Cint>0\).
    Then there exist \(a\in \intervaloo{0,1}\) and \(C_1\), \(C_2\), \(\xi>0\) such that the following holds.
    Consider the lattice \(\Lambda \Subset \Z^D\) and a Hamiltonian \(H^0_{\Lambda} + V_{\Lambda}\) as defined in Section~\ref{sec:setup} with \(\norm{h}_\infty\), \(\norm{v}_\infty \le \Cint\), \(v_x\) of range \(R \in \N\), and \(v_x\) relatively \(a\)-bounded w.r.t.~\(H^0_\Lambda\) in the sense~\eqref{eq:relative-boundedness-assumption}.
    Then the Gibbs state \(\rho_\Lambda\) at any inverse temperature \(\beta \in (0, \infty)\) satisfies
    \begin{equation}
        \label{eq:thm-DoC}
        \begin{aligned}
            \Alignindent
            \abs{
                \trace{A \, B \, \rho_{\Lambda}} - \trace{A \, \rho_{\Lambda}} \trace{B \, \rho_{\Lambda}}
            }
            \\&\le
            C_1 \, \norm{A} \, \norm{B} \, \Exp[\big]{C_2 \, (\abs{X}+\abs{Y})}\, \Exp[\big]{- \dist{X,Y}/\xi}
        \end{aligned}
    \end{equation}
    for all \(\conn\)-connected sets \(X\), \(Y\subset \Lambda\) and observables \(A \in \alg_X\) and \(B \in \alg_Y\).
\end{theorem}

We stress that all the constants \(a\), \(C_1\), \(C_2\) and \(\xi\) in the above theorem are \emph{independent} of the lattice \(\Lambda\) and in particular its size.
Moreover, they are also independent of the inverse temperature \(\beta\) and hence the correlation length \(\xi\) is uniformly bounded away from zero and infinity, indicating the absence of any kind of phase transition.
In particular, we have decay of correlations, uniformly for all temperatures.

We now state our second main result, local indistinguishability (LI) uniformly in temperature.

\begin{theorem}[Local indistinguishability]
    \label{thm:LI}
    Let \(D\), \(q\), \(R\in \N\) and \(\Cint>0\).
    Then there exist \(a\in \intervaloo{0,1}\) and \(C_1\), \(C_2\), \(\xi_{\mathrm{LI}}>0\) such that the following holds.
    Consider the lattice \(\Lambda \Subset \Z^D\) and a Hamiltonian \(H^0_{\Lambda} + V_{\Lambda}\) as defined in Section~\ref{sec:setup} with \(\norm{h}_\infty\), \(\norm{v}_\infty \le \Cint\), \(v_x\) of range \(R \in \N\), and \(v_x\) relatively \(a\)-bounded w.r.t.~\(H^0_\Lambda\) in the sense~\eqref{eq:relative-boundedness-assumption}.
    Moreover, let \(\Lambda'\subset \Lambda\) and denote with \(H_{\Lambda'}\) the Hamiltonian restricted to \(\Lambda'\).
    Then the Gibbs states \(\rho_\Lambda\) and \(\rho_{\Lambda'}\) at any inverse temperature \(\beta \in (0, \infty)\) satisfy
    \begin{equation}
        \label{eq:thm-LI}
        \abs{
            \trace{B \, \rho_{\Lambda}}
            - \trace{B \, \rho_{\Lambda'}}
        }
        \le
        C_1 \, \norm{B} \, \Exp[\big]{C_2 \, \abs{Y}}\, \Exp[\big]{- \dist{Y, \Lambda\setminus \Lambda'}/\xi_{\mathrm{LI}}}
    \end{equation}
    for all \(\conn\)-connected sets \(Y\subset \Lambda\) and observables \(B \in \alg_Y\).
\end{theorem}

As mentioned in the introduction, while one could obtain LI from DoC using quantum belief propagation~\cite{CMTW2023}, the resulting bound would scale exponentially in~\(\beta\).
Conversely, one can easily recover DoC from LI, i.e.~LI is a stronger property than DoC.

\begin{remark}[LI implies DoC]
    \label{rem:LI-implies-DoC}
    \Cref{thm:LI} also holds for unions \(Y=Y_1\cup Y_2\), where \(Y_1\) and \(Y_2\) are \(\conn\)-connected sets but \(Y\) is not.
    In particular for the case that \(B = B_1 \, B_2\) with \(B_i\in \alg_{Y_i}\) is a product observable, the proof is straightforwardly adjusted%
    \footnote{%
        \Cref{thm:DoC,thm:LI} can be generalized to~\(X\) and~\(Y\) having a finite number of connected components.
        However, we expect the constant \(C_1\) to grow faster than exponential in the number of connected components of~\(X\) and~\(Y\).
        We present the full proof only for connected observables, because that is the physically relevant setting and it reduces the technicalities.
    }.
    Then, one can easily recover DoC from LI as follows.
    Given \(\conn\)-connected sets \(X\),~\(Y\subset \Lambda\) and observables \(A\in \alg_X\), \(B\in \alg_Y\), we let \(\ell=\floor{\dist{X,Y}/2}-1\) and define the \(\ell\)-fattening
    \(X_\ell = \Set{z\in \Lambda \given \dist{z,X}\leq \ell}\) and analogously \(Y_\ell\).%
    \footnote{%
        Note that here, \(Y_\ell\) is the \(\ell\)-fattened version of the \(\conn\)-connected set~\(Y\), not \(Y_1\) or \(Y_2\) from before.
    }
    Then we apply LI with \(\Lambda'=X_\ell\cup Y_\ell\) to all expectation values in~\eqref{eq:thm-DoC}.
    Notice that \(H_{\Lambda'} = H_{X_\ell} + H_{Y_\ell}\) by the choice of~\(\ell\) and hence the Gibbs state \(\rho^\beta_{\Lambda'} = \rho^\beta_{X_\ell} \otimes \rho^\beta_{Y_\ell}\) factors and the correlations \(\Cov_{\rho^\beta_{\Lambda'}}(A,B) = 0\) vanish.
    The error terms introduced by LI sum up to a bound as in~\eqref{eq:thm-DoC} with different constants.
\end{remark}

A consequence of LI is that the finite volume states converge in the thermodynamic limit, and the resulting limit is a \(\beta\)-KMS state that satisfies DoC and LI\@.

\begin{corollary}[DoC and LI in the thermodynamic limit]
    Using the conditions and notations from \cref{thm:DoC,thm:LI}, the following holds in the thermodynamic limit \(\Lambda\nearrow \Z^D\).
    \begin{itemize}
        \item[(i)]
            For all~\(A\in \algloc\) the limit \(\tau_t^{\Z^D}(A) := \lim_{\Lambda\nearrow\Z^D} \tau_t^{\Lambda}(A)\) exists and defines a strongly continuous one-parameter family of automorphisms on~\(\alg\).
        \item[(ii)]
            For all~\(A\in \algloc\) the limit \(
                \omega_{\Z^D}^\beta(A)
                :=
                \lim_{\Lambda\nearrow\Z^D} \trace{A \, \rho_\Lambda^\beta}
            \) exists and defines a state, that is a normalized positive linear functional, on~\(\alg\).
            Moreover, the state~\(\omega_{\Z^D}^\beta\) satisfies LI as in \cref{thm:LI} with~\(\Lambda\) replaced by~\(\Z^D\) and for all finite~\(\Lambda'\Subset\Z^D\).
        \item[(iii)]
            The limit state~\(\omega_{\Z^D}^\beta\) satisfies DoC as in~\cref{thm:DoC} with the same constants.
        \item[(iv)]
            The state \(\omega_{\Z^D}^\beta\) is a \((\tau^{\Z^D},\beta)\)-KMS state.
    \end{itemize}
\end{corollary}

We conjecture that the state \(\omega_{\Z^D}^\beta\) obtained in this way is the only thermodynamic limit \(\beta\)-KMS state.

\begin{proof}
    First, note that all limits can be understood along any increasing exhausting sequence~\((\Lambda_n)_n\), i.e.~\(\Lambda_n\subset \Lambda_{n+1}\) and for every \(X\subset \Z^D\) there exist \(n\in \N\) such that \(X\subset \Lambda_n\), and the limits do not depend on the sequence.

    The convergence statement~(i) is a standard result from the literature involving Lieb-Robinson bounds, see e.g.~\cite[Theorem~3.5]{NSY2019}.

    For~(ii), note that \((\trace{A \, \rho_{\Lambda_n}^\beta})\) is a Cauchy sequence by \cref{thm:LI}, which has a unique limit.
    In this way, one obtains a state on~\(\algloc\) that can uniquely be extended to~\(\alg\) by the Hahn-Banach theorem.
    To obtain LI for~\(\omega_{\Z^D}^\beta\), we choose a sequence \(\Lambda_n \nearrow \Z^D\) with \(\Lambda'\subset \Lambda_1\).
    Then
    \begin{equation*}
        \begin{aligned}
            \abs[\big]{
                \omega_{\Z^D}^\beta(B)
                - \trace{B \, \rho_{\Lambda'}^\beta}
            }
            &\leq
            \limsup_{n \to \infty}
            \paren[\Big]{
                \abs[\big]{
                    \omega_{\Z^D}^\beta(B)
                    - \trace{B \, \rho_{\Lambda_n}^\beta}
                }
                +
                \abs[\big]{
                    \trace{B \, \rho_{\Lambda_n}^\beta}
                    - \trace{B \, \rho_{\Lambda'}^\beta}
                }
            }
            \\&\leq
            C_1 \, \norm{B} \, \Exp[\big]{C_2 \, \abs{Y}}\, \Exp[\big]{- \dist{Y, \Z^D\setminus \Lambda'}/\xi_{\mathrm{LI}}}
        \end{aligned}
    \end{equation*}
    because the first term vanishes in the limit \(\Lambda_n\nearrow\Z^D\) and the second is uniformly bounded.

    Since the constants in \cref{thm:DoC} are independent of \(\Lambda\), (iii) follows immediately from the convergence of the states.

    Finally, (iv) follows from the convergence of the dynamics and states by standard results~\cite[Proposition~5.3.23]{BR1981}.
\end{proof}

\section{Proofs}
\label{sec:proofs}

Both results, \cref{thm:DoC,thm:LI}, are proven using the same overall strategy.
We first prove \cref{thm:DoC} in three main steps.
First, in Section~\ref{sec:analytic}, we rewrite the exponential of the Hamiltonian via an inclusion-exclusion argument and estimate the terms in the decomposition exploiting analyticity of \(\C \ni z \mapsto z \, v_x\) for every \(x \in \Lambda\), as done in~\cite{Yarotsky2006}.
Next, in Section~\ref{sec:cluster}, we adapt the probabilistic cluster expansion technique from~\cite{AC2024} (see also~\cite{griffiths1970concavity,aizenman1982geometric}) to the quantum setting in order to (algebraically) unravel cancellations in the covariance.
The idea behind the underlying \emph{swapping trick} is explained in \cref{subsec:swapping}.
The remaining terms then naturally carry ratios of partition functions of \(H_S\) and \(H_S^0\) for some \(S \subset \Lambda\), which are then estimated in Section~\ref{sec:partfun}.
Finally, the arguments provided in Sections~\ref{sec:analytic}--\ref{sec:partfun} are combined in Section~\ref{sec:pfmain}, where we complete the proof of Theorem~\ref{thm:DoC}, additionally using some combinatorial percolation estimate.
We then explain the necessary modifications to prove \cref{thm:LI} in \cref{sec:proof-LI}.

\subsection{Proof strategy and the swapping trick}
\label{subsec:swapping}

We briefly describe at a high level the idea of the crucial swapping trick in the case of DoC\@.
We start by rewriting the truncated correlation function as follows
\begin{equation*}
    \trace{A \, B \, \rho^\beta}
    - \trace{A \, \rho^\beta} \, \trace{B \, \rho^\beta}
    =
    \frac{
        \trace{A \, B \, \e^{-\beta H}} \, \trace{\e^{-\beta H}}
        - \trace{A \, \e^{-\beta H}} \, \trace{B \, \e^{-\beta H}}
    }{
        \trace{\e^{-\beta H}}^2
    }
    .
\end{equation*}
We cluster expand each of the four terms in the numerator; call them \((AB),(1),(A),(B)\), respectively.
The expansion expresses the term \((AB)\) as a sum over connected clusters, with those distinguished that intersect the support of \(A\), the support of \(B\), or both.
We similarly distinguish special clusters in the expansions of the terms \((A)\) and \((B)\).
Next, we write out the product \((AB)\times (1)\) as a sum over the connected components, which we call \emph{superclusters}.
Superclusters are obtained by collating two overlapping clusters in the two expansions that make up \((AB)\times (1)\); see Definition~\ref{defn:supercluster} for their formal definition.
We then derive an analogous supercluster representation for the product \((A)\times (B)\).

Now we can perform the swapping trick.
Supercluster configurations contributing to \((AB)\times (1)\) \emph{in which \(A\) and \(B\) lie in different superclusters} can be relabelled to match exactly all the terms in the supercluster expansion of \((A)\times (B)\).
Therefore, the truncated correlation function is \emph{exactly equal} to the sum over supercluster configurations in which \(A\) and \(B\) lie in the same supercluster.
This is formalized in Theorem~\ref{thm:sum-over-superclusters-with-A-B-in-different-clusters-vanish}.
A similar swapping trick can be performed in the LI case; see Theorem~\ref{thm:LI_swap}.

We see that the swapping trick is an exact algebraic equality that takes care of the truncation part of the truncated correlation function.

\subsection{Analyticity bound}
\label{sec:analytic}
Throughout the proof, we consider both the inverse temperature \(\beta\) and the underlying lattice \(\Lambda \Subset \Z^D\) to be fixed.
In particular, every trace is understood to be taken over~\(\alg_\Lambda\).
To start, we introduce the following notations: For any \(M\subset \Lambda\) we define their \emph{interior} and \emph{closure} (relative to \(\Lambda\)) as
\begin{equation*}
    \interior{M}
    :=
    \Set[\big]{x\in M \given \suppvshift{x}\subset \Lambda}
    \qquadtext{and}
    \closure{M}
    :=
    \Lambda \cap \bigcup_{x\in M} \suppvshift{x}
    ,
\end{equation*}
such that, in particular,
\begin{equation*}
    V_\Lambda
    =
    \sum_{x\in \interior{\Lambda}} v_x
    \qquadtext{and}
    \sum_{x\in M} v_x \in \alg_{\closure{M}} \quad\text{for all \(M\subset \interior{\Lambda}\)}
    .
\end{equation*}

We now follow the expansion from~\cite{Yarotsky2006}.
First, note that for any function~\(f\) defined on \(\powerset(\Z^D)\) and every \(\Gamma \Subset \Z^D\) it holds that
\begin{equation*}
    f(\Gamma)
    =
    \sumstack{I\subset \Gamma} \sumstack{M\subset I} (-1)^{\abs{I}-\abs{M}} \, f(M)
    ,
\end{equation*}
because on the right hand side all terms except \(f(\Gamma)\) cancel exactly.
We now apply this to
\begin{equation*}
    f(M)
    =
    \Exp[\bigg]{
        -\beta \, \paren[\Big]{
            H^0_\Lambda + \sumstack[lr]{x\in M} v_x
        }
    }
\end{equation*}
to obtain
\begin{equation*}
    \e^{-\beta H_\Lambda}
    =
    f(\interior{\Lambda})
    =
    \sum_{I\subset \interior{\Lambda}} T_I
    \qquadtext{with}
    T_I
    =
    \sum_{M\subset I} (-1)^{\abs{I}-\abs{M}}
    \, f(M)
    .
\end{equation*}
In \(T_I\), the second term in the exponential is supported on \(\closure{M} \subset \closure{I}\) and thus commutes with \(H^0_{\Lambda\setminus \closure{I}} = H^0_\Lambda - H^0_{\closure{I}}\), which can thus be factored out of the exponential.
Hence,
\begin{equation*}
    T_I = \e^{-\beta H^0_{\Lambda\setminus \closure{I}}} \, T^{\closure{I}}_I
    \qquadtext{with}
    T^{I_2}_{I_1}
    =
    \sum_{M\subset I_1} (-1)^{\abs{I_1}-\abs{M}}
    \, \e^{-\beta \, \paren[\big]{H^0_{I_2} + \sum_{x\in M} v_x}}
    \in
    \alg_{I_2 \cup \closure{I_1}}
    .
\end{equation*}
With this notation, \(T_I = T^\Lambda_I\).

\begin{lemma} \label{lem:TInorm}
    Using the above notations, it holds that
    \begin{equation*}
        \norm{T^{\closure{I}}_I} \leq (2a)^{\abs{I}}
        .
    \end{equation*}
\end{lemma}

\begin{proof}
    For \(M\subset I\subset \interior{\Lambda}\) consider the function
    \begin{equation*}
        \C^{\abs{M}} \to \alg_{\closure{I}},
        \quad
        z = (z_x)_{x\in M}
        \mapsto
        g_{M}(z)
        =
        H^0_{\closure{I}}
        +
        \sumstack[lr]{x\in M} z_x \, v_x
        ,
    \end{equation*}
    which agrees with \(H^0_{\closure{I}} + \sum_{x\in M} v_x\) for \(z=(1,\dotsc,1)\) and is analytic.
    Moreover, by~\eqref{eq:relative-boundedness-assumption}
    \begin{equation*}
        \abs[\Big]{
            \innerp[\Big]{
                \psi
            }{
                \sumstack[lr]{x\in M} z_x \, v_x
                \, \psi
            }
        }
        \leq
        \sum_{x\in M}
        \, \abs{z_x}
        \, \frac{a}{\abs{\suppv}} \, \innerp[\big]{\psi}{H^0_{\suppvshift{x}} \, \psi}
        \leq
        \max_{x\in M}
        \, \abs{z_x}
        \, a \, \innerp[\big]{\psi}{H^0_{\closure{I}} \, \psi}
    \end{equation*}
    for all \(\psi\in \HS_{\closure{I}}\).
    Hence, for all \(z\) such that \(\max_x \, \abs{z_x} < 1/a\) it holds that
    \begin{equation*}
        \Re \, \innerp{\psi}{ g_M(z) \, \psi} \geq 0
        ,
    \end{equation*}
    which by the Hille-Yosida theorem implies that
    \begin{equation*}
        \norm{\e^{-\beta g_M(z)}}
        \leq
        1
        .
    \end{equation*}

    Now, for \(I\subset \interior{\Lambda}\) consider the function
    \begin{equation*}
        \C^{\abs{I}} \to \alg_{\closure{I}},
        \quad
        z
        =
        (z_x)_{x\in I}
        \mapsto
        T^{\closure{I}}_I(z)
        =
        \sum_{M\subset I} (-1)^{\abs{I}-\abs{M}} \, \e^{-\beta \, g_M\paren[\big]{(z_x)_{x\in M}}}
        .
    \end{equation*}
    This function is also analytic, satisfies \(\norm{T^{\closure{I}}_I(z)} \leq 2^{\abs{I}}\) if \(\max_{x\in I} \, \abs{z_x} < 1/a\) by the above bounds, and \(T^{\closure{I}}_I\paren[\big]{(1)_{x\in I}} = T^{\closure{I}}_I\).
    Moreover, if \(z_y = 0\) for some \(y\in I\), then \(T^{\closure{I}}_I(z) = 0\).
    To understand this, note that for any \(M\subset I \setminus\List{y}\)
    \begin{equation*}
        g_{M\cup\List{y}}\paren[\big]{(z_x)_{x\in M\cup\List{y}}}
        =
        g_{M}\paren[\big]{(z_x)_{x\in M}}
    \end{equation*}
    and due to the different signs in \(T^{\closure{I}}_I(z)\) all terms cancel.
    The result follows by applying~\cite[Lemma~2]{Yarotsky2006}.
\end{proof}

\subsection{Cluster expansion}
\label{sec:cluster}

Using Yarotsky's decomposition from last section, for any \(\Omega\subset \Lambda\) and \(O\in \alg_\Omega\) we can decompose
\begin{align*}
    \frac{\trace[\big]{O \, \e^{-\beta \, H_\Lambda}}}{Z^0_\Lambda}
    &=
    \sum_{I \subset \interior{\Lambda}} \frac{\trace[\big]{O \, T_I}}{Z^0_\Lambda}
    \\&=
    \sum_{I \subset \interior{\Lambda}}
    \frac{
        \trace[\big]{
            O
            \, T^{\closure{I}}_I
            \, \e^{-\beta H^0_{\Omega\setminus\closure{I}}}
            \, \e^{-\beta H^0_{\Lambda\setminus \smash[b]{(\closure{I} \cup \Omega)}}}
        }
    }{
        \trace[\big]{
            \e^{-\beta H^0_{\closure{I}}}
            \, \e^{-\beta H^0_{\Omega\setminus\closure{I}}}
            \, \e^{-\beta H^0_{\Lambda\setminus \smash[b]{(\closure{I} \cup \Omega)}}}
        }
    }
    \\&=
    \sum_{I \subset \interior{\Lambda}} \frac{\trace[\big]{O \, T^{\closure{I}\cup \Omega}_I}}{Z^0_{\closure{I}\cup \Omega}}
    ,
\end{align*}
where due to \(O\), the trace only factorizes between \(\alg_{\closure{I} \cup \Omega}\) and \(\alg_{\Lambda\setminus(\closure{I} \cup \Omega)}\).
For \(O=\unit\), one can choose \(\Omega=\emptyset\) and obtain
\begin{equation*}
    \frac{\trace[\big]{\e^{-\beta \, H_\Lambda}}}{Z^0_\Lambda}
    =
    \sum_{I \in \interior{\Lambda}} \frac{\trace[\big]{T^{\closure{I}}_I}}{Z^0_{\closure{I}}}
    .
\end{equation*}

At this point, we would like to interpret the subset \(I\) that is being summed over on the right-hand side of the equations above as a \enquote{configuration} that is produced by the partition function.
This interpretation allows us to phrase the methods we will later use in this paper in the language of statistical physics.
Accordingly, we will also give the following definition to link our quantities to concepts from statistical physics.

\begin{definition}[Configuration, Support, and Weight]
    \label{def:confsuppweight}
    We defined a configuration \(I\) to be any subset of \(\interior{\Lambda}\).
    The support of the configuration is equivalent to the subset itself.
    The weight of the configuration is \(\trace{T^{\closure{I}}_I} / Z^0_{\closure{I}}\).
\end{definition}

\begin{remark}[On the \(h_x\)'s being on-site operators]
    In implementing the decomposition above and using the concepts introduced in Definition~\ref{def:confsuppweight}, we crucially used that the operators \(h_x\) act only on individual sites.
    This is necessary to introduce the notion of \enquote{weights} and use that the weight of union \(I \cup J\) factorizes, if \(I\) and \(J\) are not \emph{\(\conn\)-connected} to each other (see Definition~\ref{def:connected-to-each-other} later).
\end{remark}

\begin{remark}[Comparison to the probabilistic setting]
    The only notion presented above that would distinguish our present setup from standard statistical physics is the fact that the weight associated to each configuration is not necessarily positive.
    Moreover, we remark that the notion of \enquote{support} is only mentioned to complete the analogy, but not used in the argument.
\end{remark}

A critical idea in the paper~\cite{Yarotsky2006} and statistical physics in general is the idea of cluster and the compatibility between clusters.
Indeed, one can very strongly characterize properties of the Gibbs measures in statistical physics if one knows that configurations can be broken apart into distinct clusters and the weight of a configuration is a simple product of the weight of each individual cluster.

\begin{definition}[\(\conn\)-connectedness]
    \label{def:connected-to-each-other}
    Let \(x\) and \(y\) be two points of \(\Lambda\).
    We say that \(x\)~and~\(y\) are \emph{\(\conn\)-connected} if \(\suppvshift{x} \cap \suppvshift{y} \ne \emptyset\).
    Similarly, we say that two sets \(I_1\), \(I_2\subset \Lambda\) are \emph{\(\conn\)-connected to each other}, if there exist \(\conn\)-connected points \(x_1\in I_1\) and \(x_2\in I_2\).
\end{definition}

Given this \(\conn\)-connectedness, we define clusters as \(\conn\)-connected sets as in \cref{def:connected-sets-setup}.

\begin{definition}[Cluster]
    \label{def:connected-sets}
    A subset \(I\subset \Lambda\) is a \emph{cluster} or equivalently an \emph{\(\conn\)-connected set}, if for every two points \(x\), \(y\in I\) there exists a sequence of points \(z_1,\ldots,z_m\in I\) such that \(z_i\) is \(\conn\)-connected to \(z_{i+1}\) for \(i \in \{1,\ldots, m-1\}\) and \(x\) is \(\conn\)-connected to \(z_1\) and \(y\) is \(\conn\)-connected to \(z_m\).
\end{definition}

The following Lemma shows that our configurations and weights satisfy the fundamental properties of a cluster expansion.

\begin{lemma}
    \label{lem:split-two-clusters}
    Let \(I_1\), \(I_2\subset \interior{\Lambda}\) and \(\Omega_1\), \(\Omega_2\subset \Lambda\).
    Assume that \(I_1\cup \Omega_1\) and \(I_2\cup \Omega_2\) are not \(\conn\)-connected to each other.
    Then, for every \(O_1\in \alg_{\Omega_1}\) and \(O_2\in \alg_{\Omega_2}\), we have that
    \begin{equation}
        \label{eq:clustersplit}
        \frac{
            \trace{
                O_1 \, O_2 \, T^{\closure{I_1 \cup I_2}\cup \Omega_1\cup \Omega_2}_{I_1 \cup I_2}
            }
        }{Z^0_{\closure{I_1 \cup I_2}\cup \Omega_1\cup \Omega_2}}
        =
        \frac{
            \trace{
                O_1 \, T^{\closure{I_1}\cup \Omega_1}_{I_1}
            }
        }{Z^0_{\closure{I_1}\cup \Omega_1}}
        \, \frac{
            \trace{
                O_2 \, T^{\closure{I_2}\cup \Omega_2}_{I_2}
            }
        }{Z^0_{\closure{I_2}\cup \Omega_2}}
        .
    \end{equation}
\end{lemma}

\begin{proof}
    We first write
    \begin{equation*}
        T^{\closure{I_1 \cup I_2}\cup \Omega_1\cup \Omega_2}_{I_1 \cup I_2}
        =
        T^{\closure{I_1 \cup I_2}}_{I_1 \cup I_2}
        \, \e^{-\beta H^0_{(\Omega_1\cup \Omega_2)\setminus \closure{I_1 \cup I_2}}}
        .
    \end{equation*}
    Since \(I_1\) and \(I_2\) are not \(\conn\)-connected to each other, we must have that \(\closure{I_1} \cap \closure{I_2} = \emptyset\).
    Thus, \(\closure{I_1 \cup I_2} = {\closure{I_1} \cup \closure{I_2}}\) is a disjoint union.
    Furthermore, \(I_1\) and \(I_2\) are disjoint and any subset \(M \subset I_1 \cup I_2\) can be uniquely decomposed as \(M= M_1 \cup M_2\), where \(M_1\subset I_1\) and \(M_2\cup I_2\).
    As a consequence, we see that
    \begin{equation*}
        \begin{aligned}
            T_{I_1 \cup I_2}^{\closure{I_1 \cup I_2}}
            &=
            \sum_{M \subset I_1 \cup I_2}
            (-1)^{\abs{I_1 \cup I_2} - \abs{M}}
            \, \Exp[\bigg]{
                -\beta \, \paren[\Big]{
                    H^0_{\closure{I_1 \cup I_2}}
                    + \sum_{x \in M} v_x
                }
            }
            \\&=
            \begin{aligned}[t]
                \sum_{M \subset I_1 \cup I_2}
                (-1)^{\abs{I_1}- \abs{M_1}}
                \, (-1)^{\abs{I_2}- \abs{M_2}}
                \, & \Exp[\bigg]{
                    -\beta \, \paren[\Big]{
                        H^0_{\closure{I_1}}
                        + \sum_{x \in M_1} v_x
                    }
                }
                \\&\times
                \Exp[\bigg]{
                    -\beta \, \paren[\Big]{
                        H^0_{\closure{I_2}}
                        + \sum_{x \in M_2} v_x
                    }
                }
            \end{aligned}
            \\&=
            T^{\closure{{I_1}}}_{I_1} \, T^{\closure{{I_2}}}_{I_2}
            .
        \end{aligned}
    \end{equation*}
    In the second line, we first decomposed the subset \(M = M_1 \cup M_2\) as above.
    The operators \(\sum_{x \in M_1} v_x\) and \(H^0_{\closure{I_1}}\) are in \(\alg_{\closure{I_1}}\) while \(\sum_{x \in M_2} v_x\) and \(H^0_{\closure{I_2}}\) are in \(\alg_{\closure{I_2}}\).
    Thus, the operators \(H^0_{\closure{I_1}} + \sum_{x \in M_1} v_x\) and \(H^0_{\closure{I_2}} + \sum_{x \in M_2} v_x\) commute with each other.
    Since \(H^0\) has only on-site terms, we also have
    \begin{align*}
        \, \e^{-\beta H^0_{(\Omega_1\cup \Omega_2)\setminus \closure{I_1 \cup I_2}}}
        =
        \e^{-\beta \, H^0_{\Omega_1 \setminus \closure{I_1}}}
        \, \e^{-\beta \, H^0_{\Omega_2 \setminus \closure{I_2}}}
    \shortintertext{and}
        \e^{-\beta \, H^0_{\closure{I_1 \cup I_2}\cup \Omega_1\cup \Omega_2}}
        =
        \e^{-\beta \, H^0_{\closure{I_1}\cup \Omega_1}}
        \, \e^{-\beta \, H^0_{\closure{I_2}\cup \Omega_2}}
        .
    \end{align*}
    Using factorization of the trace between \(\closure{I_1}\cup \Omega_1\) and \(\closure{I_2}\cup \Omega_2\) gives~\eqref{eq:clustersplit}.
\end{proof}

With these consequences in hand, we can now get to the crux of the matter.
\begin{definition}[Supercluster]
    \label{defn:supercluster}
    Let \(I\) and \(J\) be two configurations.
    A supercluster decomposition of \(I \cup J \cup X \cup Y\) is a decomposition
    \begin{equation*}
        I \cup J \cup X \cup Y = S_1 \cup S_2 \cup \ldots \cup S_n
        ,
    \end{equation*}
    where the sets \(S_i\) are all maximally \(\conn\)-connected subsets, i.e.\ the \(S_i\) are \(\conn\)-connected sets that are not \(\conn\)-connected to each other.
\end{definition}

We can provide configuration expansions for \(
    \trace{\e^{-\beta \, H_{\Lambda}}} / Z^0_{\Lambda}
\), \(
    \trace{A \, B \, \e^{-\beta \, H_\Lambda}} / Z^0_{\Lambda}
\) as well as \(
    \trace{A \, \e^{-\beta \, H_{\Lambda}}} / Z^0_{\Lambda}
\) and \(
    \trace{B \, \e^{-\beta \, H_{\Lambda}}} / Z^0_{\Lambda}
\).
The product of the configuration decompositions of \(
    \trace{\e^{-\beta \, H_{\Lambda}}} / Z^0_{\Lambda}
\) and \(
    \trace{A \, B \, \e^{-\beta \, H_\Lambda}} / Z^0_{\Lambda}
\)
will lead to a supercluster decomposition which will, in some instances, cancel with appropriate terms that appear in a supercluster decomposition involving \(
    \trace{A \, \e^{-\beta \, H_{\Lambda}}} / Z^0_{\Lambda}
\) and \(
    \trace{B \, \e^{-\beta \, H_{\Lambda}}} / Z^0_{\Lambda}
\).
This is the content of the following theorem.

\begin{theorem}
    \label{thm:sum-over-superclusters-with-A-B-in-different-clusters-vanish}
    Let \(X\) and \(Y\subset \Lambda\) each be an \(\conn\)-connected set and \(A\in \alg_X\) and \(B\in \alg_Y\).
    Let \(I\) be a configuration produced in the decomposition of \(
        \trace{\e^{-\beta \, H_{\Lambda}}} / Z^0_{\Lambda}
    \) and \(J\) be a configuration produced in the decomposition of \(
        \trace{A \, B \, \e^{-\beta \, H_{\Lambda}}} / Z^0_{\Lambda}
    \).
    We consider the supercluster decomposition of \(I \cup J \cup X \cup Y\) and let \(E(X,Y)\) be the event that the associated supercluster decomposition does not contain \(X\) and \(Y\) in the same supercluster.
    As a shorthand, we write such terms in the product as
    \begin{equation*}
        \frac{\trace{\e^{-\beta \, H_{\Lambda}}}}{Z^0_{\Lambda}}
        \, \frac{\trace{A \, B \, \e^{-\beta \, H_{\Lambda}}}}{Z^0_{\Lambda}}
        \bigg|_{E(X,Y)}
        .
    \end{equation*}
    We consider a similar supercluster decomposition for
    \begin{equation*}
        \frac{\trace{A \, \e^{-\beta \, H_{\Lambda}}}}{Z^0_{\Lambda}}
        \, \frac{\trace{B \, \e^{-\beta \, H_{\Lambda}}}}{Z^0_{\Lambda}}
        \bigg|_{E(X,Y)}
        .
    \end{equation*}
    The two quantities above are equal.
\end{theorem}

\begin{proof}
    Consider a supercluster expansion for
    \begin{equation*}
        \frac{\trace{\e^{-\beta \, H_{\Lambda}}}}{Z^0_{\Lambda}}
        \, \frac{\trace{A \, B \, \e^{-\beta \, H_{\Lambda}}}}{Z^0_{\Lambda}}
        \bigg|_{E(X,Y)}
        .
    \end{equation*}
    This supercluster expansion can be written as \(S_1 \cup S_2 \cup \ldots \cup S_m\) and without loss of generality we can assume that \(X \subset S_1\) and \(Y \subset S_2\).
    We now define
    \begin{align*}
        I'
        &=
        (J\cap S_1)
        \cup
        \paren[\big]{
            I \cap (S_2 \cup S_3 \cup \dotsb \cup S_m)
        }
    \shortintertext{and}
        J'
        &=
        (I\cap S_1)
        \cup
        \paren[\big]{
            J \cap (S_2 \cup S_3 \cup \dotsb \cup S_m)
        }
        .
    \end{align*}
    On an intuitive level, we switch the parts of the configurations \(I\) and \(J\) that are part of the first supercluster \(S_1\), which contains \(X\).

    We remark that the supercluster decomposition of \(X \cup Y \cup I' \cup J'\) is still \(S_1 \cup S_2 \cup \ldots \cup S_m\).
    The union of \(I' \cup J'\) must be the same as the union of \(I \cup J\) by our construction and, thus, we could not have changed the maximally \(\conn\)-connected subsets that appear in the decomposition.
    Furthermore, \(Y\) will still be contained in the supercluster \(S_2\) and \(X\) is contained in the supercluster \(S_1\).
    As a result, we see that this construction is an involution.
    We only have to check that
    \begin{equation}
        \label{eq:superclusterswap}
        \frac{\trace{T^{\closure{I}}_I}}{Z^0_{\closure{I}}}
        \, \frac{\trace{A \, B \, T^{\closure{J}\cup X \cup Y}_J}}{Z^0_{\closure{J}\cup X \cup Y}}
        =
        \frac{\trace{A \, T^{\closure{I'}\cup X}_{I'}}}{Z^0_{\closure{I'}\cup X}}
        \, \frac{\trace{B \, T^{\closure{J'}\cup Y}_{J'}}}{Z^0_{\closure{J'}\cup Y}}
        .
    \end{equation}
    Since the superclusters \(S_1,\ S_2,\ldots,\ S_m\) are mutually not \(\conn\)-connected to each other and \(X\subset S_1\) and \(Y\subset S_2\), we can apply \cref{lem:split-two-clusters} to each of the fractions, yielding
    \begin{align*}
        \frac{\trace{T^{\closure{I}}_I}}{Z^0_{\closure{I}}}
        &=
        \prod_{k=1}^m
        \frac{
            \trace{T^{\closure{{I \cap S_k}}}_{I \cap S_k}}
        }{
            Z^0_{\closure{I \cap S_k}}
        }
        ,\\
        \frac{\trace{A \, B \, T^{\closure{J}\cup X\cup Y}_J}}{Z^0_{\closure{J}\cup X\cup Y}}
        &=
        \frac{
            \trace{A \, T^{\closure{{J \cap S_1}}\cup X}_{J \cap S_1}}
        }{
            Z^0_{\closure{J \cap S_1}\cup X}
        }
        \, \frac{
            \trace{B \, T^{\closure{{J \cap S_2}}\cup Y}_{J \cap S_2}}
        }{
            Z^0_{\closure{J \cap S_2}\cup Y}
        }
        \, \prod_{k=3}^m
        \frac{
            \trace{T^{\closure{{J \cap S_k}}}_{J \cap S_k}}
        }{
            Z^0_{\closure{J \cap S_k}}
        }
    \intertext{for the left-hand-side and}
        \frac{\trace{A \, T^{\closure{I'}\cup X}_{I'}}}{Z^0_{\closure{I'}\cup X}}
        &=
        \frac{
            \trace{A \, T^{\closure{{I' \cap S_1}}\cup X}_{I' \cap S_1}}
        }{
            Z^0_{\closure{I' \cap S_1}\cup X}
        }
        \, \frac{
            \trace{T^{\closure{{I' \cap S_2}}}_{I' \cap S_2}}
        }{
            Z^0_{\closure{I' \cap S_2}}
        }
        \, \prod_{k=3}^m
        \frac{
            \trace{T^{\closure{{I' \cap S_k}}}_{I' \cap S_k}}
        }{
            Z^0_{\closure{I' \cap S_k}}
        }
        \\&=
        \frac{
            \trace{A \, T^{\closure{{J \cap S_1}}\cup X}_{J \cap S_1}}
        }{
            Z^0_{\closure{J \cap S_1}\cup X}
        }
        \, \frac{
            \trace{T^{\closure{{I \cap S_2}}}_{I \cap S_2}}
        }{
            Z^0_{\closure{I \cap S_2}}
        }
        \, \prod_{k=3}^m
        \frac{
            \trace{T^{\closure{{I \cap S_k}}}_{I \cap S_k}}
        }{
            Z^0_{\closure{I \cap S_k}}
        }
        ,\\
        \frac{\trace{B \, T^{\closure{J'}\cup Y}_{J'}}}{Z^0_{\closure{J'}\cup Y}}
        &=
        \frac{
            \trace{T^{\closure{{J' \cap S_1}}}_{J' \cap S_1}}
        }{
            Z^0_{\closure{J' \cap S_1}}
        }
        \, \frac{
            \trace{B \, T^{\closure{{J' \cap S_2}}\cup Y}_{J' \cap S_2}}
        }{
            Z^0_{\closure{J' \cap S_2}\cup Y}
        }
        \, \prod_{k=3}^m
        \frac{
            \trace{T^{\closure{{J' \cap S_k}}}_{J' \cap S_k}}
        }{
            Z^0_{\closure{J' \cap S_k}}
        }
        \\&=
        \frac{
            \trace{T^{\closure{{I \cap S_1}}}_{I \cap S_1}}
        }{
            Z^0_{\closure{I \cap S_1}}
        }
        \, \frac{
            \trace{B \, T^{\closure{{J \cap S_2}}\cup Y}_{J \cap S_2}}
        }{
            Z^0_{\closure{J \cap S_2}\cup Y}
        }
        \, \prod_{k=3}^m
        \frac{
            \trace{T^{\closure{{J \cap S_k}}}_{J \cap S_k}}
        }{
            Z^0_{\closure{J \cap S_k}}
        }
    \end{align*}
    for the right-hand-side.
    Observing that these terms are the same proves~\eqref{eq:superclusterswap} and shows that we have equality of the supercluster expansions on the event \(E(X,Y)\).
\end{proof}

At this point, it remains to control the product of the cluster expansions when \(X\) and \(Y\) are part of the same supercluster.
Namely, what we will do is to first fix a given supercluster \(S\) that connects \(X\) and \(Y\) and consider all pairs of configurations \(I\) and \(J\) such that \(S\) is a supercluster of \(X \cup Y \cup I \cup J\).
We will then sum over all of these pairs of configurations such that \(I\) and \(J\) are superclusters of this configuration.
More formally, what we will actually do is fix the intersection \(I \cap S\) and \(J \cap S\), and sum over all pairs of clusters with such a fixed value of \(I \cap S\) and \(J \cap S\).
We have the following Lemma.

\begin{lemma} \label{lem:superclust}
    Fix subsets \(I_0\), \(J_0\subset \interior{\Lambda}\) such that \(S_0 := I_0 \cup J_0 \cup X \cup Y\) is an \(\conn\)-connected set.
    Denote with \(\calC_{I_0,J_0}\) the set of pairs \((I,J)\), with \(I\), \(J\subset \interior{\Lambda}\), whose supercluster decomposition \(I\cup J\cup X\cup Y = S_0 \cup S_1\cup \dotsb\cup S_m\) contains \(S_0\) and satisfies \(I\cap S_0 = I_0\) and \(J\cap S_0 = J_0\).
    Then, we have that
    \begin{align}
        \label{eq:coefficients}
        \sum_{(I,J)\in \calC_{I_0,J_0}}
        \frac{\trace{T^{\closure{I}}_I}}{Z^0_{\closure{I}}}
        \, \frac{\trace{A \, B \, T^{\closure{J}\cup X \cup Y}_J}}{Z^0_{\closure{J}\cup X \cup Y}}
        =
        \frac{\trace{T^{\closure{I_0}}_{I_0}}}{Z^0_{\closure{I_0}}}
        \, \frac{\trace{A \, B \, T^{\closure{J_0}\cup X \cup Y}_{J_0}}}{Z^0_{\closure{J_0}\cup X \cup Y}}
        \, \paren[\bigg]{
            \frac{Z_{\Lambda\setminus \closure{S_0}}}{Z^0_{\Lambda\setminus \closure{S_0}}}
        }^2
    \shortintertext{and}
        \label{eq:lem-X-Y-in-same-super-cluster-second-term}
        \sum_{(I,J)\in \calC_{I_0,J_0}}
        \frac{\trace{A \, T^{\closure{I}\cup X}_I}}{Z^0_{\closure{I}\cup X}}
        \, \frac{\trace{B \, T^{\closure{J}\cup Y}_J}}{Z^0_{\closure{J}\cup Y}}
        =
        \frac{\trace{A \, T^{\closure{I_0}\cup X}_{I_0}}}{Z^0_{\closure{I_0}\cup X}}
        \, \frac{\trace{B \, T^{\closure{J_0}\cup Y}_{J_0}}}{Z^0_{\closure{J_0}\cup Y}}
        \, \paren[\bigg]{
            \frac{Z_{\Lambda\setminus \closure{S_0}}}{Z^0_{\Lambda\setminus \closure{S_0}}}
        }^2
        .
    \end{align}
\end{lemma}

\begin{proof}
    First, note that \(S_0\) and \(S_1\cup \dotsb\cup S_m\) are not \(\conn\)-connected to each other, and \(X\), \(Y\subset S_0\).
    Hence, we can decompose \(\calC_{I_0,J_0}\) as
    \begin{align*}
        \calC_{I_0,J_0}
        &=
        \Set[\Big]{
            (I_0\cup I', J_0\cup J')
            \given
            I',\ J'\subset \interior{\Lambda}
            \text{ such that }
            \closure{I'}\cap \closure{S_0} = \emptyset
            \text{ and }
            \closure{J'}\cap \closure{S_0} = \emptyset
        }
        \\&=
        \Set[\Big]{
            (I_0\cup I', J_0\cup J')
            \given
            I',\ J'
            \subset
            \interior{\Lambda}\setminus \closure{\closure{S_0}}
            = \interior{\paren[\big]{\Lambda\setminus \closure{S_0}}}
        }
        .
    \end{align*}
    Using this and \cref{lem:split-two-clusters}, we obtain
    \begin{align*}
        \Alignindent
        \sum_{(I,J)\in \calC_{I_0,J_0}}
        \frac{\trace{T^{\closure{I}}_I}}{Z^0_{\closure{I}}}
        \, \frac{\trace{A \, B \, T^{\closure{J}\cup X \cup Y}_J}}{Z^0_{\closure{J}\cup X \cup Y}}
        \\&=
        \frac{\trace{T^{\closure{I_0}}_{I_0}}}{Z^0_{\closure{I_0}}}
        \, \frac{\trace{A \, B \, T^{\closure{J_0}\cup X \cup Y}_{J_0}}}{Z^0_{\closure{J_0}\cup X \cup Y}}
        \sum_{I',J'\subset \interior{\paren[\big]{\Lambda\setminus \closure{S_0}}}}
        \frac{\trace{T^{\closure{I'}}_{I'}}}{Z^0_{\closure{I'}}}
        \, \frac{\trace{T^{\closure{J'}}_{J'}}}{Z^0_{\closure{J'}}}
        .
    \end{align*}
    As in the decomposition of \cref{sec:analytic}, we have
    \begin{equation*}
        \sum_{I'\subset \interior{\paren[\big]{\Lambda\setminus \closure{S_0}}}}
        \frac{\trace{T^{\closure{I'}}_{I'}}}{Z^0_{\closure{I'}}}
        =
        \sum_{I'\subset \interior{\paren[\big]{\Lambda\setminus \closure{S_0}}}}
        \frac{\trace{T^{\Lambda\setminus \closure{S_0}}_{I'}}}{Z^0_{\Lambda\setminus \closure{S_0}}}
        =
        \frac{\trace{\e^{-\beta H_{\Lambda\setminus \closure{S_0}}}}}{Z^0_{\Lambda\setminus \closure{S_0}}}
        =
        \frac{Z_{\Lambda\setminus \closure{S_0}}}{Z^0_{\Lambda\setminus \closure{S_0}}}
    \end{equation*}
    and~\eqref{eq:coefficients} follows.
    Equation~\eqref{eq:lem-X-Y-in-same-super-cluster-second-term} follows in exactly the same way.
\end{proof}

\subsection{Ratios of partition functions}
\label{sec:partfun}
In the previous section, we extracted important cancellations between various terms arising in the cluster expansion (see Theorem~\ref{thm:sum-over-superclusters-with-A-B-in-different-clusters-vanish}).
The remaining terms (see Lemma~\ref{lem:superclust}) involve ratios of partition functions associated to the Hamiltonians \(H_\Lambda\) and \(H^0_\Lambda\), which we control with the aid of the following lemma.

\begin{lemma} \label{lem:ratiopartition}
    Let \(D\), \(q\), \(R\in \N\) and \(a\in \intervaloo{0,1}\).
    Then there exist \(C>0\) such that the following holds.
    Consider the lattice \(\Lambda \Subset \Z^D\) with local Hilbert space \(\HS_x \equiv \C^q\), \(q \ge 2\), at each \(x \in \Lambda\), and a Hamiltonian \(H^0_{\Lambda} + V_{\Lambda}\) as defined in Section~\ref{sec:setup} with \(v_x\) relatively \(a\)-bounded w.r.t.~\(H^0_\Lambda\) in the sense~\eqref{eq:relative-boundedness-assumption} and \(v_x \leq 0\).
    Then for all \(\beta\in \intervaloo{0,\infty}\) and all \(\conn\)-connected sets \(S\subset \Lambda\)
    \begin{equation*}
        \frac{Z^0_\Lambda}{Z_\Lambda}
        \, \frac{Z_{\Lambda\setminus \closure{S}}}{Z^0_{\Lambda\setminus \closure{S}}}
        =
        \frac{ Z_{\Lambda\setminus \closure{S}} \, Z^0_{\closure{S}}}{Z_\Lambda}
        \leq
        1
        .
    \end{equation*}
\end{lemma}

\begin{proof}
    The proof relies on the following basic trace inequality.
    For any two \(n\times n\) self-adjoint matrices \(H\) and \(H'\) satisfying the operator inequality \(H\leq H'\), we have the monotonicity
    \begin{equation}
        \label{eq:trace_monotonic}
        \trace[\big]{\e^{H}} \leq \trace[\big]{\e^{H'}}
        .
    \end{equation}
    To see that~\eqref{eq:trace_monotonic} holds, write \(\lambda_1\leq \dotsb \leq\lambda_n\) and \(\lambda_1'\leq \dotsb \leq\lambda_n'\) for the ordered eigenvalues of \(H\) and \(H'\), respectively.
    By the min-max principle and \(H\leq H'\), we have \(\lambda_i\leq \lambda_i'\) for all \(i\) and so the spectral theorem implies~\eqref{eq:trace_monotonic}.

    Note, that the equality in the statement just follows by the factorization of \(Z^0_\Lambda\) and \(Z^0_{\Lambda\setminus \closure{S}}\).
    And by monotonicity we find
    \begin{equation}
        \label{eq:lemma_Zratios_proof}
        Z_{\Lambda\setminus \closure{S}} \, Z^0_{\closure{S}}
        =
        \trace[\Big]{
            \e^{-\beta \, (H^0_{\Lambda} + V_{\Lambda\setminus \closure{S}})}
        }
        \leq
        \trace[\Big]{
            \e^{-\beta \, H_\Lambda}
        }
        =
        Z_\Lambda
        ,
    \end{equation}
    since
    \begin{equation*}
        H^0_{\Lambda} + V_{\Lambda\setminus \closure{S}}
        \geq
        H^0_{\Lambda} + V_{\Lambda}
        = H_\Lambda
    \end{equation*}
    as we add terms \(v_x \leq 0\).
    And the statement follows.
\end{proof}

\subsection{Proof of Theorem~\ref{thm:DoC}}
\label{sec:pfmain}

We can now put everything together and prove \cref{thm:DoC}.
By the decomposition introduced in \cref{sec:analytic}, we have
\begin{align*}
    \label{eq:correlation}
    \Alignindent
    \abs{
        \trace{A \, B \, \rho_{\Lambda}}
        - \trace{A \, \rho_{\Lambda}}
        \, \trace{B \, \rho_{\Lambda}}
    }
    \\&=
    \paren[\bigg]{\frac{Z^0_\Lambda}{Z_\Lambda}}^2
    \, \abs[\bigg]{
        \, \frac{\trace{\e^{-\beta H_{\Lambda}}}}{Z^0_\Lambda}
        \, \frac{\trace{A \, B \, \e^{-\beta H_{\Lambda}}}}{Z^0_\Lambda}
        - \frac{\trace{A \, \e^{-\beta H_{\Lambda}}}}{Z^0_\Lambda}
        \, \frac{\trace{B \, \e^{-\beta H_{\Lambda}}}}{Z^0_\Lambda}
    }
    \\&=
    \paren[\bigg]{\frac{Z^0_\Lambda}{Z_\Lambda}}^2
    \abs[\bigg]{
        \sum_{I,J\subset \interior{\Lambda}}
        \, \frac{\trace{T_I^{\closure{I}}}}{Z^0_{\closure{I}}}
        \, \frac{\trace{A \, B \, T_J^{\closure{J}\cup X\cup Y}}}{Z^0_{\closure{J}\cup X\cup Y}}
        - \frac{\trace{A \, T_I^{\closure{I}\cup X}}}{Z^0_{\closure{I}\cup X}}
        \, \frac{\trace{B \, T_J^{\closure{J}\cup Y}}}{Z^0_{\closure{J}\cup Y}}
    }
    .
\end{align*}
And by \cref{thm:sum-over-superclusters-with-A-B-in-different-clusters-vanish}, the part of the sum where the supercluster decomposition of \(I\cup J\cup X\cup Y\) has \(X\) and \(Y\) in two disjoint clusters vanishes.
What remains is a sum over \(I\), \(J\subset \interior{\Lambda}\) such that the supercluster decomposition \(I\cup J\cup X\cup Y = S_0\cup S_1\cup \dotsm \cup S_m\) satisfies \(X\), \(Y\subset S_0\).
In this cluster \(S_0\), we still have the freedom to choose \(I_0 = I \cap S_0\) and \(J_0 = J \cap S_0\).
Then, using the notation from \cref{lem:superclust}, we obtain the upper bound
\begin{align*}
    \Alignindent
    \paren[\Big]{\fnfrac{Z^0_\Lambda}{Z_\Lambda}}^2
    \, \sumstack[l]{S_0\subset \Lambda\suchthat\\ \textup{\(S_0\) \(\conn\)-connected}\\ X,Y\subset S_0}
    \, \sumstack{I_0,J_0\subset \interior{\Lambda}\suchthat\\I_0\cup J_0\cup X\cup Y = S_0}
    \, \abs[\bigg]{
        \sumstack[r]{(I,J)\in \calC_{I_0,J_0}}
        \, \fnfrac{\trace{T_I^{\closure{I}}}}{Z^0_{\closure{I}}}
        \, \fnfrac{\trace{A \, B \, T_J^{\closure{J}\cup X\cup Y}}}{Z^0_{\closure{J}\cup X\cup Y}}
        - \fnfrac{\trace{A \, T_I^{\closure{I}\cup X}}}{Z^0_{\closure{I}\cup X}}
        \, \fnfrac{\trace{B \, T_J^{\closure{J}\cup Y}}}{Z^0_{\closure{J}\cup Y}}
    }
    \\&\leq
    \paren[\Big]{\fnfrac{Z^0_\Lambda}{Z_\Lambda}}^2
    \, \sumstack[lr]{S_0\subset \Lambda\suchthat\\ \textup{\(S_0\) \(\conn\)-connected}\\ X,Y\subset S_0}
    \, \paren[\bigg]{\fnfrac{Z_{\Lambda\setminus \closure{S_0}}}{Z^0_{\Lambda\setminus \closure{S_0}}}}^2
    \, \sumstack[l]{I_0,J_0\subset \interior{\Lambda}\suchthat\\\mathclap{I_0\cup J_0\cup X\cup Y = S_0}}
    \abs[\bigg]{
        \fnfrac{\trace{T^{\closure{I_0}}_{I_0}}}{Z^0_{\closure{I_0}}}
        \, \fnfrac{\trace{A \, B \, T^{\closure{J_0}\cup X \cup Y}_{J_0}}}{Z^0_{\closure{J_0}\cup X \cup Y}}
    }
    + \abs[\bigg]{
        \fnfrac{\trace{A \, T^{\closure{I_0}\cup X}_{I_0}}}{Z^0_{\closure{I_0}\cup X}}
        \, \fnfrac{\trace{B \, T^{\closure{J_0}\cup Y}_{J_0}}}{Z^0_{\closure{J_0}\cup Y}}
    }
    .
\end{align*}

For each \(I\), \(\Omega\subset S_0\) and \(O\in \alg_\Omega\), we now use \cref{lem:TInorm} to control the two summands in the last line as
\begin{equation*}
    \abs[\bigg]{
        \frac{\trace{O \, T^{\closure{I}\cup \Omega}_{I}}}{Z^0_{\closure{I}\cup \Omega}}
    }
    =
    \abs[\bigg]{
        \frac{
            \trace_{\closure{I}\cup \Omega}{O \, T^{\closure{I}\cup \Omega}_{I}}
        }{
            \trace_{\closure{I}}{\e^{-\beta H^0_{\closure{I}}}}
            \, \trace_{\Omega\setminus\closure{{I}}}{\e^{-\beta H^0_{\Omega\setminus\closure{{I}}}}}
        }
    }
    \leq
    \norm{O}
    \, \paren[\big]{2 \, a \, q^{(2R+1)^D}}^{\abs{I}}
    .
\end{equation*}
Here, we explicitly took the trace only over \(\alg_{\closure{I}\cup \Omega}\) and then estimated
\begin{equation*}
    \abs{
        \trace_{\closure{I}\cup \Omega}{O \, T^{\closure{I}\cup \Omega}_{I}}
    }
    \leq
    \norm{O}
    \, \norm[\big]{
        T^{\closure{I}}_{I} \, \e^{-\beta H^0_{\Omega\setminus\closure{I}}}
    }_{\Tr_{\closure{I}\cup \Omega}}
    \leq
    \norm{O}
    \, \norm[\big]{
        T^{\closure{I}}_{I}
    }_{\Tr_{\closure{I}}}
    \, \norm[\big]{
        \e^{-\beta H^0_{\Omega\setminus\closure{I}}}
    }_{\Tr_{\Omega\setminus\closure{I}}}
\end{equation*}
and
\begin{equation*}
    \norm[\big]{
        T^{\closure{I}}_{I}
    }_{\Tr_{\closure{I}}}
    \leq
    \norm[\big]{
        T^{\closure{I}}_{I}
    }
    \, \norm{\unit_{\closure{I}}}_{\Tr_{\closure{I}}}
    \leq
    (2a)^{\abs{I}}
    \, q^{\abs{\closure{I}}}
    \leq
    (2a)^{\abs{I}}
    \, q^{(2R+1)^D \, \abs{I}}
\end{equation*}
by \cref{lem:TInorm}.
Putting everything together, we obtain the bound
\begin{align*}
    \Alignindent
    \abs{
        \trace{A \, B \, \rho_{\Lambda}}
        - \trace{A \, \rho_{\Lambda}}
        \, \trace{B \, \rho_{\Lambda}}
    }
    \\&\leq
    2
    \, \norm{A}
    \, \norm{B}
    \, \sumstack[lr]{S_0\subset \Lambda\suchthat\\ \textup{\(S_0\) \(\conn\)-connected}\\ X,Y\subset S_0}
    \, \paren[\bigg]{
        \fnfrac{Z^0_\Lambda}{Z_\Lambda}
        \, \fnfrac{Z_{\Lambda\setminus \closure{S_0}}}{Z^0_{\Lambda\setminus \closure{S_0}}}
    }^2
    \, \sumstack[lr]{I_0,J_0\subset \interior{\Lambda}\suchthat\\I_0\cup J_0\cup X\cup Y = S_0}
    \, \paren[\big]{2 \, a \, q^{(2R+1)^D}}^{\abs{I_0}+\abs{J_0}}
    .
\end{align*}
The ratio of partition functions in the squared parenthesis is upper bounded by \(1\) due to \cref{lem:ratiopartition}, recalling that \(v_x \le 0\) w.l.o.g., as explained in Section~\ref{sec:setup}

Next, to control the inner sum over \(I_0\) and \(J_0\), we abbreviate \(p := 2 \, a \, q^{(2R+1)^D}\) and \(M:=S_0\setminus (X \cupdot Y)\).
Then, after replacing \(X \to X \cap \interior{\Lambda}\) and \(Y \to Y \cap \interior{\Lambda}\), we have
\begin{align*}
    \sumstack[r]{I_0,J_0\subset \interior{\Lambda}\suchthat\\I_0\cup J_0\cup X\cup Y = S_0}
    \, p^{\abs{I_0}+\abs{J_0}}
    &=
    \sumstack[l]{I_M,J_M\subset M\suchthat\\I_M\cup J_M = M}
    \, \sumstack[r]{I_X,J_X\subset X\\I_Y,J_Y\subset Y}
    \, p^{\abs{I_M \cupdot I_X \cupdot I_Y}+\abs{J_M \cupdot J_X \cupdot J_Y}}
    \\&=
    \sumstack[lr]{I_M,J_M\subset M\suchthat\\I_M\cup J_M = M}
    \, p^{\abs{I_M}+\abs{J_M}}
    \, (1+p)^{2\abs{X}+2\abs{Y}}
    \\&=
    \sumstack[l]{I_M\subset M}
    \sumstack[r]{\tilde{J}\subset I_M}
    \, p^{\abs{I_M}+\abs{\tilde{J}\cupdot M\setminus I_M}}
    \, (1+p)^{2\abs{X}+2\abs{Y}}
    \\&=
    p^{\abs{M}}
    \, \sumstack[l]{I_M\subset M}
    \, \sumstack[r]{\tilde{J}\subset I_M}
    \, p^{\abs{\tilde{J}}}
    \, (1+p)^{2\abs{X}+2\abs{Y}}
    \\&=
    (2p)^{\abs{M}}
    \, (1+p)^{\abs{M}}
    \, (1+p)^{2\abs{X}+2\abs{Y}}
    ,
\end{align*}
where we repeatedly used
\begin{equation*}
    \sum_{E\subset F} \, p^{\abs{E}}
    =
    \sum_{n=0}^{\abs{F}} \, \binom{\abs{F}}{n} \, p^n \, 1^{\abs{F}-n}
    =
    (1+p)^{\abs{F}}
    .
\end{equation*}
In this way, we obtain
\begin{equation}
    \label{eq:almostdone}
    \begin{aligned}
        \Alignindent
        \abs{
            \trace{A \, B \, \rho_{\Lambda}}
            - \trace{A \, \rho_{\Lambda}}
            \, \trace{B \, \rho_{\Lambda}}
        }
        \\&\leq
        2
        \, \norm{A}
        \, \norm{B}
        \, (1 + p)^{2 \abs{X} + 2 \abs{Y}}
        \, \sumstack[lr]{S_0\subset \Lambda\suchthat \\ \textup{\(S_0\) \(\conn\)-connected}\\ X,Y\subset S_0}
        \, \paren[\big]{2p \, (1+p)}^{\abs{S_0 \setminus (X \cupdot Y)}}
        \\&\leq
        2
        \, \norm{A}
        \, \norm{B}
        \, \paren*{\frac{1}{2} + \frac{1}{2p}}^{ \abs{X} + \abs{Y}}
        \, \sumstack[lr]{S_0\subset \Lambda\suchthat \\ \textup{\(S_0\) \(\conn\)-connected}\\ X,Y\subset S_0}
        \, \paren[\big]{2p \, (1+p)}^{\abs{S_0}}
        .
    \end{aligned}
\end{equation}

Finally, since \(S_0\) in~\eqref{eq:almostdone} is an \(\conn\)-connected set and contains \(X\) and \(Y\), it must have at least \(k = \dist{X,Y}/(2R)\) sites, since there must be \(z_0, z_1,\ldots, z_k\in S_0\) with \(z_0\in X\), \(z_k\in Y\) and \(\dist{z_i,z_{i+1}} \leq 2R\).
Then, by the following lemma, whose proof is given at the end of the section, there are at most \(C^k\) such \(S_0\) with \(\abs{S_0}=k\).
\begin{lemma} \label{lem:counting}
    Let \(D\in \N\) and \(R\in \N\).
    Then there exists a constant \(C>0\) such that for any \(v_1\in \Z^D\) the number \(\conn\)-connected subsets \(S\subset \Z^D\) containing \(v_1\in S\) and satisfying \(\abs{S}=k\in \N\) is bounded by~\(C^k\).
\end{lemma}

Therefore, we find that
\begin{align*}
    \Alignindent
    \abs{
        \trace{A \, B \, \rho_{\Lambda}}
        -\trace{A \, \rho_{\Lambda}}
        \, \trace{B \, \rho_{\Lambda}}}
    \\&\le
    2
    \, \norm{A}
    \, \norm{B}
    \, \paren*{\frac{1}{2} + \frac{1}{2p}}^{ \abs{X} + \abs{Y}}
    \sum_{k= \dist{X,Y}/(2 \conn)}^{\infty} \paren[\big]{2p \, (1+p) \, C}^{k}
    \\&\le
    C
    \, \norm{A}
    \, \norm{B}
    \, \paren*{\frac{1}{2} + \frac{1}{2p}}^{ \abs{X} + \abs{Y}}
    \paren[\big]{2p \, (1+p) \, C}^{\dist{X,Y}/(2 \conn)}
    .
\end{align*}
Here, in the last step, we chose \(a\) such that \(2p \, (1+p) \, C < 1\) (recall the shorthand notation \(p = 2 \, a \, q^{(2R + 1)^D}\)).
This completes the proof of Theorem~\ref{thm:DoC}.
\qed

\bigskip

It remains to give the proof of Lemma~\ref{lem:counting}.

\begin{proof}[Proof of Lemma~\ref{lem:counting}]
    We will use the following algorithm to count the number of \(\conn\)-connected subsets of \(\Z^D\) that contain a specific point~\(v_1\).
    Therefore, fix any well-ordering \(v_1 \prec v_2 \prec v_3 \prec \dotsb\) on \(\Z^D\).

    Given any \(\conn\)-connected set \(S\subset \Z^D\) of size \(k\) that contains \(v_1\), we construct an ordered list \(Q=(q_1,q_2,\dotsc,q_k)\) according to the following algorithm.
    We begin the algorithm by setting \(n=1\), \(i=1\) and \(q_1 = v_1\).
    In each step of the algorithm we do the following:
    If the set
    \begin{equation*}
        M := S \cap \suppvshift{q_i} \setminus \List{q_1,\dotsc,q_n}
    \end{equation*}
    of vertices, which are in \(S\) and \(\conn\)-connected to \(q_i\) but not yet in \(Q\), is empty, we increase \(i\) by one.
    Otherwise, we set the next element in \(Q\) to the lowest element of \(M\) according to the chosen well-ordering \(q_{n+1} = \min_{\prec} M\) and then increase \(n\) by one.
    We stop the algorithm after \(2k-2\) steps, when \(n=i=k\) and \(S = \List{q_1,\dotsc,q_k}\).

    Conversely, we can now count the number of \(\conn\)-connected sets \(S\subset \Z^D\) containing \(v_1\) and having \(\abs{S}=k\) sites by counting the number of ways this algorithm could run.
    Therefore, we observe that there are less than \(\binom{2(k-1)}{k-1}\) possible ways of splitting between the two branches in the algorithm.
    Moreover, in the second branch, one has less than \(\abs[\big]{\suppvshift{q_i}} \leq (2R+1)^D\) possibilities of choosing \(q_{n+1}\).
    Hence, in total there can be at most
    \begin{equation*}
        \binom{2 \, (k-1)}{k-1} \, (2R+1)^{D(k-1)}
        \leq
        \paren[\big]{2 \, \e \, (2R+1)^{D} }^{k-1}
    \end{equation*}
    \(\conn\)-connected sets of size \(k\) that contain \(v_1\).
\end{proof}

\subsection{Proof of Theorem~\ref{thm:LI}}
\label{sec:proof-LI}

We now adjust the proof for local indistinguishability.
Therefore, fix \(\Lambda'\subset \Lambda\).
We use the same notation as before for \(H_\Lambda^0\) and \(H_\Lambda\) and all derived quantities.
Additionally, we define an interaction with support in \(\Lambda'\), by defining local terms
\begin{equation*}
    \tilde{v}_x =
    \begin{cases*}
        v_x & if \(B_R(x)\subset \Lambda'\)
        \\
        0   & otherwise.
    \end{cases*}
\end{equation*}
All symbols derived with \(\tilde{v}_x\) instead of \(v_x\) are also denoted with a tilde.
In this way, \(\tilde{V}_\Lambda = V_{\Lambda'}\) and \(\tilde{H}_\Lambda = H_\Lambda^0 + \tilde{V}_\Lambda\).
Since the Hamiltonian \(H_\Lambda^0\) only has on-site contributions, the exponential \(
    \e^{-\beta\tilde{H}_\Lambda}
    =
    \e^{-\beta H_{\Lambda'}} \, \e^{-\beta H^0_{\Lambda\setminus\Lambda'}}
\) factors, and \(\trace{A \, \rho_{\Lambda'}} = \trace{A \, \tilde{\rho}_\Lambda}\) for all \(A\in \alg_{\Lambda'}\).
Hence, we only need to compare \(\tilde{\rho}_\Lambda\) and \(\rho_\Lambda\) in the following.

Clearly, also the interaction \(\tilde{v}_x\) satisfies the assumptions from \cref{sec:setup}, and thus the derived quantities \(\Tt_{I_1}^{I_2}\) satisfy \cref{lem:TInorm,lem:split-two-clusters,lem:ratiopartition}.
We are left to adjust the cluster expansion.
Therefore, we write
\begin{align*}
    \label{eq:LI}
    \Alignindent
    \abs{
        \trace{B \, \rho_{\Lambda}}
        - \trace{B \, \tilde\rho_{\Lambda}}
    }
    \\&=
    \frac{(Z^0_\Lambda)^2}{Z_\Lambda \, \Zt_\Lambda}
    \, \abs[\bigg]{
        \, \frac{\trace{B \, \e^{-\beta H_{\Lambda}}}}{Z^0_\Lambda}
        \, \frac{\trace{\e^{-\beta \tilde{H}_{\Lambda}}}}{Z^0_\Lambda}
        - \frac{\trace{B \, \e^{-\beta \tilde{H}_{\Lambda}}}}{Z^0_\Lambda}
        \, \frac{\trace{\e^{-\beta H_{\Lambda}}}}{Z^0_\Lambda}
    }
    \\&=
    \frac{(Z^0_\Lambda)^2}{Z_\Lambda \, \Zt_\Lambda}
    \, \abs[\bigg]{
        \sumstack[r]{I,J\subset \interior{\Lambda}}
        \, \frac{\trace{B \, T_I^{\closure{I}\cup Y}}}{Z^0_{\closure{I}\cup Y}}
        \, \frac{\trace{\Tt_J^{\closure{J}}}}{Z^0_{\closure{J}}}
        - \frac{\trace{B \, \Tt_I^{\closure{I}\cup Y}}}{Z^0_{\closure{I}\cup Y}}
        \, \frac{\trace{T_J^{\closure{J}}}}{Z^0_{\closure{J}}}
    }
\end{align*}
for any \(B\in \alg_Y\) and then observe the equivalent statement to \cref{thm:sum-over-superclusters-with-A-B-in-different-clusters-vanish}.

\begin{theorem}
    \label{thm:LI_swap}
    Let \(Y\subset \Lambda'\) be an \(\conn\)-connected set and \(B\in \alg_Y\).
    Let \(I\) be a configuration produced in the decomposition of \(
        \trace{B \, \e^{-\beta \, H_{\Lambda}}} / Z^0_{\Lambda}
    \) and \(J\) be a configuration produced in the decomposition of \(
        \trace{\e^{-\beta \, \tilde{H}_{\Lambda}}} / Z^0_{\Lambda}
    \).
    We consider the supercluster decomposition of \(I \cup J \cup Y\) and let \(E(Y,\interior{\Lambda'})\) be the event that the associated supercluster decomposition contains \(Y\) in a cluster that itself is contained in \(\interior{\Lambda'}\).
    As a shorthand, we write such terms in the product as
    \begin{equation*}
        \frac{\trace{B \, \e^{-\beta \, H_{\Lambda}}}}{Z^0_{\Lambda}}
        \, \frac{\trace{\e^{-\beta \, \tilde{H}_{\Lambda}}}}{Z^0_{\Lambda}}
        \bigg|_{E(Y,\interior{\Lambda'})}
        .
    \end{equation*}
    We consider the same supercluster decomposition for
    \begin{equation*}
        \frac{\trace{B \, \e^{-\beta \, \tilde{H}_{\Lambda}}}}{Z^0_{\Lambda}}
        \, \frac{\trace{\e^{-\beta \, H_{\Lambda}}}}{Z^0_{\Lambda}}
        \bigg|_{E(Y,\interior{\Lambda'})}
        .
    \end{equation*}
    The two quantities are equal.
\end{theorem}

\begin{proof}
    We consider a supercluster expansion \(I\cup J\cup Y = S_1\cup \dotsb \cup S_m\).
    Without loss of generality, we assume \(Y\subset S_1\).
    If \(S_1\subset \interior{\Lambda'}\), then \(
        \Tt_{I\cap S_1}^{\closure{{I \cap S_1}}\cup Y}
        =
        T_{I\cap S_1}^{\closure{{I \cap S_1}}\cup Y}
    \) and \(
        \Tt_{J\cap S_1}^{\closure{{J \cap S_1}}}
        =
        T_{J\cap S_1}^{\closure{{J \cap S_1}}}
    \) and there is another configuration
    \begin{align*}
        I' &= (I\cap S_1) \cup \paren[\big]{
            J \cap (S_2 \cup S_3 \cup \dotsb \cup S_m)
        },
        \\
        J' &= (J\cap S_1) \cup \paren[\big]{
            I \cap (S_2 \cup S_3 \cup \dotsb \cup S_m)
        }
    \end{align*}
    with the same supercluster decomposition, such that
    \begin{align*}
        \Alignindent
        \frac{\trace{B \, T_I^{\closure{I}\cup Y}}}{Z^0_{\closure{I}\cup Y}}
        \, \frac{\trace{\Tt_J^{\closure{J}}}}{Z^0_{\closure{J}}}
        \\&=
        \frac{
            \trace{B \, T^{\closure{{I \cap S_1}}\cup Y}_{I \cap S_1}}
        }{
            Z^0_{\closure{I \cap S_1}\cup Y}
        }
        \, \frac{
            \trace{\Tt^{\closure{{J \cap S_1}}}_{J \cap S_1}}
        }{
            Z^0_{\closure{J \cap S_1}}
        }
        \, \prod_{k=2}^m
        \frac{
            \trace{T^{\closure{{I \cap S_k}}}_{I \cap S_k}}
        }{
            Z^0_{\closure{I \cap S_k}}
        }
        \, \frac{
            \trace{\Tt^{\closure{{J \cap S_k}}}_{J \cap S_k}}
        }{
            Z^0_{\closure{J \cap S_k}}
        }
        \\&=
        \frac{
            \trace{B \, \Tt^{\closure{{I' \cap S_1}}\cup Y}_{I' \cap S_1}}
        }{
            Z^0_{\closure{I' \cap S_1}\cup Y}
        }
        \, \frac{
            \trace{T^{\closure{{J' \cap S_1}}}_{J' \cap S_1}}
        }{
            Z^0_{\closure{J' \cap S_1}}
        }
        \, \prod_{k=2}^m
        \frac{
            \trace{T^{\closure{{J' \cap S_k}}}_{J' \cap S_k}}
        }{
            Z^0_{\closure{J' \cap S_k}}
        }
        \, \frac{
            \trace{\Tt^{\closure{{I' \cap S_k}}}_{I' \cap S_k}}
        }{
            Z^0_{\closure{I' \cap S_k}}
        }
        \\&=
        \frac{\trace{B \, \Tt_{I'}^{\closure{I'}\cup Y}}}{Z^0_{\closure{I'}\cup Y}}
        \, \frac{\trace{T_{J'}^{\closure{J'}}}}{Z^0_{\closure{J'}}}
        .
    \end{align*}
    This concludes the proof.
\end{proof}

Hence, as for DoC, what remains is a sum over \(I\),~\(J\subset \interior{\Lambda}\) such that the supercluster decomposition \(I\cup J\cup Y = S_0 \cup \dotsb \cup S_m\) satisfies \(Y\subset S_0\) and \(S_0\cap (\Lambda\setminus \interior{\Lambda'}) \neq \emptyset\).
We now apply \cref{lem:superclust} with \(X=\emptyset\) and \(A=\unit\), to obtain the upper bound
\begin{equation*}
    \fnfrac{(Z^0_\Lambda)^2}{Z_\Lambda \, \Zt_\Lambda}
    \, \sumstack[lr]{S_0\subset \Lambda\suchthat\\ \textup{\(S_0\) \(\conn\)-connected}\\ Y\subset S_0\\S_0\cap (\Lambda\setminus \interior{\Lambda'}) \neq \emptyset}
    \, \fnfrac{Z_{\Lambda\setminus \closure{S_0}} \, \Zt_{\Lambda\setminus \closure{S_0}}}{(Z_{\Lambda\setminus \closure{S_0}}^0)^2}
    \, \sumstack[l]{I_0,J_0\subset \interior{\Lambda}\suchthat\\\mathclap{I_0\cup J_0\cup Y = S_0}}
    \, \abs[\bigg]{
        \fnfrac{\trace{B \, T_{I_0}^{\closure{I_0}\cup Y}}}{Z^0_{\closure{I_0}\cup Y}}
        \, \fnfrac{\trace{\Tt_{J_0}^{\closure{J_0}}}}{Z^0_{\closure{J_0}}}
    }
    + \abs[\bigg]{
        \fnfrac{\trace{B \, \Tt_{I_0}^{\closure{I_0}\cup Y}}}{Z^0_{\closure{I_0}\cup Y}}
        \, \fnfrac{\trace{T_{J_0}^{\closure{J_0}}}}{Z^0_{\closure{J_0}}}
    }
    .
\end{equation*}
Following the arguments in the proof of \cref{thm:DoC}, we obtain
\begin{align*}
    \Alignindent
    \abs{
        \trace{B \, \rho_{\Lambda}}
        - \trace{B \, \tilde\rho_{\Lambda}}
    }
    \\&\leq
    2
    \, \norm{B}
    \, \sumstack[lr]{S_0\subset \Lambda\suchthat\\ \textup{\(S_0\) \(\conn\)-connected}\\ Y\subset S_0\\S_0\cap (\Lambda\setminus \interior{\Lambda'}) \neq \emptyset}
    \, \fnfrac{(Z^0_\Lambda)^2}{Z_\Lambda \, \Zt_\Lambda}
    \, \fnfrac{Z_{\Lambda\setminus \closure{S_0}} \, \Zt_{\Lambda\setminus \closure{S_0}}}{(Z^0_{\Lambda\setminus \closure{S_0}})^2}
    \, \sumstack[lr]{I_0,J_0\subset \interior{\Lambda}\suchthat\\I_0\cup J_0\cup Y = S_0}
    \, p^{\abs{I_0}+\abs{J_0}}
    \\&\leq
    2
    \, \norm{B}
    \, (1+p)^{2\abs{Y}}
    \, \sumstack[lr]{S_0\subset \Lambda\suchthat\\ \textup{\(S_0\) \(\conn\)-connected}\\ Y\subset S_0\\S_0\cap (\Lambda\setminus \interior{\Lambda'}) \neq \emptyset}
    \, \fnfrac{Z^0_\Lambda \, Z_{\Lambda\setminus \closure{S_0}}}{Z_\Lambda \, Z^0_{\Lambda\setminus \closure{S_0}}}
    \, \fnfrac{\Zt^0_\Lambda \, \Zt_{\Lambda\setminus \closure{S_0}}}{\Zt_\Lambda \, \Zt^0_{\Lambda\setminus \closure{S_0}}}
    \, \paren[\big]{2p \, (1+p)}^{\abs{S_0\setminus Y}}
    \\&\leq
    2
    \, \norm{B}
    \, \paren[\Big]{\fnfrac{1}{2}+\fnfrac{1}{2p}}^{\abs{Y}}
    \, \sumstack[lr]{S_0\subset \Lambda\suchthat\\ \textup{\(S_0\) \(\conn\)-connected}\\ Y\subset S_0\\S_0\cap (\Lambda\setminus \interior{\Lambda'}) \neq \emptyset}
    \, \paren[\big]{2p \, (1+p)}^{\abs{S_0}}
    \\&\leq
    2
    \, \norm{B}
    \, \paren[\Big]{\fnfrac{1}{2}+\fnfrac{1}{2p}}^{\abs{Y}}
    \, \paren[\big]{2p \, (1+p) \, C}^{\dist{Y,\Lambda\setminus \interior{\Lambda'}}/(2R)}
    ,
\end{align*}
with \(p = 2 \, a \, q^{(2R+1)^D}\) and \(2p \, (1+p) \, C < 1\) for \(a\) small enough.

\section{Local perturbations perturb locally}
\label{sec:LPPL}

In this section, we provide a result on the local perturbations perturb locally (LPPL) principle.
Compared to our main results on the more important notions of DoC and LI, the LPPL bound deteriorates exponentially as \(\beta \to \infty\).
We conjecture that for the systems we consider, the LPPL bound actually holds uniformly in temperature, but we are unable to prove it with the new method in this paper.

\begin{theorem}[Local perturbations perturb locally]
    \label{thm:LPPL}
    Let \(D\), \(q\), \(R\in \N\) and \(\Cint>0\).
    Then there exist \(a\in \intervaloo{0,1}\) and \(C_1\), \(C_2\), \(\xi_{\mathrm{LPPL}}>0\) such that the following holds.
    Consider the lattice \(\Lambda \Subset \Z^D\) and a Hamiltonian \(H^0_{\Lambda} + V_{\Lambda}\) as defined in Section~\ref{sec:setup} with \(\norm{h}_\infty\), \(\norm{v}_\infty \le \Cint\), \(v_x\) of range \(R \in \N\), and \(v_x\) relatively \(a\)-bounded w.r.t.~\(H^0_\Lambda\) in the sense~\eqref{eq:relative-boundedness-assumption}.
    Moreover, let \(X \subset \Lambda\) be an \(\conn\)-connected set, \(W\subset \alg_X\) self-adjoint and \(\Ht_\Lambda = H_\Lambda + W\).
    Then the Gibbs states \(\rho_\Lambda\) and \(\tilde\rho_{\Lambda}\) of \(H_\Lambda\) and \(\Ht_\Lambda\), respectively, at any inverse temperature \(\beta \in (0, \infty)\) satisfy
    \begin{equation}
        \abs{
            \trace{B \, \rho_{\Lambda}}
            - \trace{B \, \tilde\rho_{\Lambda}}
        }
        \le
        C_1 \, \norm{B} \, \e^{2\beta\norm{W}} \, \Exp[\big]{C_2 \, (\abs{X}+\abs{Y})} \, \Exp[\big]{- \dist{X,Y}/\xi_{\mathrm{LPPL}}}
    \end{equation}
    for all \(\conn\)-connected sets \(Y\subset \Lambda\) and observables \(B \in \alg_Y\).
\end{theorem}

Notably, this statement also holds, if the perturbed Hamiltonian \(H_\Lambda + W\) does not satisfy the conditions from \cref{sec:setup}.
If \(H_\Lambda + W\) was of the same type, applying \cref{thm:LI} for both Hamiltonians together with the triangle inequality gives a uniform-in-temperature bound \(
    2 \, C_1 \, \norm{B} \, \Exp[\big]{C_2 \, \abs{Y}} \, \Exp[\big]{- \dist{X,Y}/\xi_{\mathrm{LPPL}}}
\).

\begin{remark}
    The \(\beta\)-dependence is still better than what one would obtain from our \cref{thm:DoC} on DoC with the \emph{circle of equivalences} from~\cite{CMTW2025}.
    Similarly as in the case of LI, the problem with the circle of equivalences arises from quantum belief propagation.
    First, it produces constants that diverge as \(\beta\to\infty\).
    Second, one can only obtain stretched exponential decay because one needs to use DoC for the observable \(B\) and an approximation of the quantum belief propagation operator that lives on a suitably enlarged region \(X_r\) with \(r\) a free parameter that can be optimized.
    As a consequence, the form of LPPL that one obtains in this way from the good DoC bound Theorem~\ref{thm:DoC} through the circle of equivalences is suboptimal.
    Instead, using our cluster expansion approach, we are able to obtain LPPL with exponential decay in Theorem~\ref{thm:LPPL}, but the constants diverge as \(\beta\to\infty\).
\end{remark}

\subsection{Proof of Theorem~\ref{thm:LPPL}}
\label{sec:proof-LPPL}

We focus on the modifications necessary to prove \cref{thm:LPPL}.
We begin by adjusting the quantities from \cref{sec:analytic} to the counterparts for \(\Ht_\Lambda = H_\Lambda + W\).
Again, we denote all modified symbols with an additional tilde.
First, let
\begin{equation*}
    \ft(M)
    =
    \Exp[\bigg]{
        -\beta \, \paren[\Big]{
            H^0_\Lambda + W + \sumstack[lr]{x\in M} v_x
        }
    }
    ,
\end{equation*}
which clearly satisfies \(\e^{-\beta \Ht} = \ft(\interior{\Lambda})\).
Due to the inclusion-exclusion principle, we can decompose it as \(\ft(\interior{\Lambda}) = \sum_{I\subset \interior{\Lambda}} \Tt_I\) with
\begin{equation*}
    \Tt_I
    =
    \sum_{M\subset I} (-1)^{\abs{I}-\abs{M}}
    \, \ft(M)
\end{equation*}
and
\begin{equation}
    \label{eq:LPPL-Tt-factor-H0}
    \Tt_I
    =
    \e^{-\beta H^0_{\Lambda\setminus (\closure{I}\cup X)}}
    \, \Tt_I^{\closure{I}\cup X}
    ,
\end{equation}
where
\begin{equation*}
    \Tt^{I_2}_{I_1}
    =
    \sum_{M\subset I_1} (-1)^{\abs{I_1}-\abs{M}}
    \, \e^{-\beta \, \paren[\big]{H^0_{I_2} + W + \sum_{x\in M} v_x}}
    \in
    \alg_{I_2 \cup \closure{I_1}}
    .
\end{equation*}
Similarly to \cref{lem:TInorm}, we find
\begin{lemma} \label{lem:LPPL-TInorm}
    Using the above notations, it holds that
    \begin{equation*}
        \norm{\Tt^{\closure{I}}_I}
        \leq
        (2a)^{\abs{I}} \, \e^{\beta\norm{W}}
        .
    \end{equation*}
\end{lemma}

\begin{proof}
    The proof follows the one from \cref{lem:TInorm}, but due to the additional \(W\) in \(\tilde{g}_M(z)\), we only have
    \begin{equation*}
        \Re\,\innerp[\big]{\psi}{ \gt_M(z) \, \psi}
        \geq
        \Re\,\innerp{\psi}{ W \, \psi}
        \geq
        \inf_\phi \, \innerp{\phi}{ W \, \phi}
        .
    \end{equation*}
    Hence,
    \begin{equation*}
        \norm{\e^{-\beta \gt_M(z)}}
        \leq
        \e^{-\beta \inf_\phi \innerp{\phi}{ W \, \phi}}
        \leq
        \e^{\beta \norm{W}}
    \end{equation*}
    and the remaining arguments together with~\cite[Lemma~2]{Yarotsky2006} yield the statement.
\end{proof}

When we split the exponential into products of clusters, we need to assign one cluster to have the perturbation~\(W\).
Hence, we need to modify \cref{lem:split-two-clusters} for the perturbed system.
Notably, we choose \(\Ht^0_\Lambda = H^0_\Lambda + W\), so that \(\Zt^0_\Lambda\) includes the perturbation~\(W\).

\begin{lemma}
    \label{lem:LPPL-split-two-clusters}
    Let \(I_1\), \(I_2\subset \interior{\Lambda}\) and \(\Omega_1\), \(\Omega_2\subset \Lambda\).
    Assume that \(I_1\cup \Omega_1 \cup X\) and \(I_2\cup \Omega_2\) are not \(\conn\)-connected to each other.
    Then, for every \(O_1\in \alg_{\Omega_1}\) and \(O_2\in \alg_{\Omega_2}\), we have that
    \begin{equation}
        \frac{
            \trace{
                O_1 \, O_2 \, \Tt^{\closure{I_1 \cup I_2}\cup \Omega_1\cup \Omega_2 \cup X }_{I_1 \cup I_2}
            }
        }{\Zt^0_{\closure{I_1 \cup I_2}\cup \Omega_1\cup \Omega_2 \cup X }}
        =
        \frac{
            \trace{
                O_1 \, \Tt^{\closure{I_1}\cup \Omega_1 \cup X }_{I_1}
            }
        }{\Zt^0_{\closure{I_1}\cup \Omega_1 \cup X }}
        \, \frac{
            \trace{
                O_2 \, T^{\closure{I_2}\cup \Omega_2}_{I_2}
            }
        }{Z^0_{\closure{I_2}\cup \Omega_2}}
        .
    \end{equation}
\end{lemma}

\begin{proof}
    The proof follows exactly the proof of \cref{lem:split-two-clusters}, only that \(W\) is attributed to one term on each side, and the terms \(h_x\) for \(x\in X\) cannot be factored out.
\end{proof}

As in the proof of DoC and LI, we obtain the expansion
\begin{align*}
    \Alignindent
    \abs{
        \trace{B \, \rho_{\Lambda}}
        - \trace{B \, \tilde\rho_{\Lambda}}
    }
    \\&=
    \frac{Z^0_\Lambda \, \Zt^0_\Lambda}{Z_\Lambda \, \Zt_\Lambda}
    \, \abs[\bigg]{
        \sum_{I,J\subset \interior{\Lambda}}
        \, \frac{\trace{B \, T_I^{\closure{I}\cup Y}}}{Z^0_{\closure{I}\cup Y}}
        \, \frac{\trace{\Tt_J^{\closure{J}\cup X}}}{\Zt^0_{\closure{J}\cup X}}
        - \frac{\trace{B \, \Tt_I^{\closure{I}\cup Y\cup X}}}{\Zt^0_{\closure{I}\cup Y\cup X}}
        \, \frac{\trace{T_J^{\closure{J}}}}{Z^0_{\closure{J}}}
    }
    .
\end{align*}

Equivalently to \cref{thm:sum-over-superclusters-with-A-B-in-different-clusters-vanish}, we find that only superclusters that contain~\(X\) and~\(Y\) in the same cluster contribute.

\begin{theorem}
    Let \(X\) and \(Y\subset \Lambda\) each be an \(\conn\)-connected set and \(A\in \alg_X\) and \(B\in \alg_Y\).
    Let \(I\) be a configuration produced in the decomposition of \(
        \trace{B \, \e^{-\beta \, H_{\Lambda}}} / Z^0_{\Lambda}
    \) and \(J\) be a configuration produced in the decomposition of \(
        \trace{\e^{-\beta \, \tilde{H}_{\Lambda}}} / Z^0_{\Lambda}
    \).
    We consider the supercluster decomposition of \(I \cup J \cup X \cup Y\) and let \(E(X,Y)\) be the event that the associated supercluster decomposition does not contain \(X\) and \(Y\) in the same supercluster.
    As a shorthand, we write such terms in the product as
    \begin{equation*}
        \frac{\trace{B \, T_I^{\closure{I}\cup Y}}}{Z^0_{\closure{I}\cup Y}}
        \, \frac{\trace{\Tt_J^{\closure{J}\cup X}}}{\Zt^0_{\closure{J}\cup X}}
        \bigg|_{E(X,Y)}
        .
    \end{equation*}
    We consider a similar supercluster decomposition for
    \begin{equation*}
        \frac{\trace{B \, \Tt_I^{\closure{I}\cup Y\cup X}}}{\Zt^0_{\closure{I}\cup Y\cup X}}
        \, \frac{\trace{T_J^{\closure{J}}}}{Z^0_{\closure{J}}}
        \bigg|_{E(X,Y)}
        .
    \end{equation*}
    The two quantities above are equal.
\end{theorem}

\begin{proof}
    We consider a supercluster expansion \(I\cup J\cup Y \cup X = S_1 \cup S_2 \cup \dots \cup S_m\) for the event \(E(X,Y)\), and without loss of generality, we can assume \(X\subset S_1\) and \(Y\subset S_2\).
    We then split each cluster using \cref{lem:split-two-clusters,lem:LPPL-split-two-clusters}.
    By denoting \(S_{\mathsf{c}} = S_3\cup \dots \cup S_m\) we obtain,
    \begin{align*}
        \Alignindent
        \frac{\trace{B \, T_I^{\closure{I}\cup Y}}}{Z^0_{\closure{I}\cup Y}}
        \, \frac{\trace{\Tt_J^{\closure{J}\cup X}}}{\Zt^0_{\closure{J}\cup X}}
        \\&=
        \frac{\trace{T_{I\cap S_1}^{\closure{I\cap S_1}}}}{Z^0_{\closure{I\cap S_1}}}
        \, \frac{\trace{B \, T_{I\cap S_2}^{\closure{I\cap S_2}\cup Y}}}{Z^0_{\closure{I\cap S_2}\cup Y}}
        \, \frac{\trace{T_{I\cap S_{\mathsf{c}}}^{\closure{I\cap S_{\mathsf{c}}}}}}{Z^0_{\closure{I\cap S_{\mathsf{c}}}}}
        \, \frac{\trace{\Tt_{J\cap S_1}^{\closure{J\cap S_1}\cup X}}}{\Zt^0_{\closure{J\cap S_1}\cup X}}
        \, \frac{\trace{T_{J\cap S_2}^{\closure{J\cap S_2}\cup X}}}{Z^0_{\closure{J\cap S_2}\cup X}}
        \, \frac{\trace{T_{J\cap S_{\mathsf{c}}}^{\closure{J\cap S_{\mathsf{c}}}\cup X}}}{Z^0_{\closure{J\cap S_{\mathsf{c}}}\cup X}}
        \\&=
        \frac{\trace{T_{J'\cap S_1}^{\closure{J'\cap S_1}}}}{Z^0_{\closure{J'\cap S_1}}}
        \, \frac{\trace{B \, T_{I'\cap S_2}^{\closure{I'\cap S_2}\cup Y}}}{Z^0_{\closure{I'\cap S_2}\cup Y}}
        \, \frac{\trace{T_{I'\cap S_{\mathsf{c}}}^{\closure{I'\cap S_{\mathsf{c}}}}}}{Z^0_{\closure{I'\cap S_{\mathsf{c}}}}}
        \, \frac{\trace{\Tt_{I'\cap S_1}^{\closure{I'\cap S_1}\cup X}}}{\Zt^0_{\closure{I'\cap S_1}\cup X}}
        \, \frac{\trace{T_{J'\cap S_2}^{\closure{J'\cap S_2}\cup X}}}{Z^0_{\closure{J'\cap S_2}\cup X}}
        \, \frac{\trace{T_{J'\cap S_{\mathsf{c}}}^{\closure{J'\cap S_{\mathsf{c}}}\cup X}}}{Z^0_{\closure{J'\cap S_{\mathsf{c}}}\cup X}}
        \\&=
        \frac{\trace{B \, \Tt_{I'}^{\closure{I'}\cup Y \cup X}}}{\Zt^0_{\closure{I'}\cup Y \cup X}}
        \, \frac{\trace{T_{J'}^{\closure{J'}}}}{Z^0_{\closure{J'}}}
        ,
    \end{align*}
    where
    \begin{align*}
        I' &= (J\cap S_1) \cup \paren[\big]{
            I \cap (S_2 \cup S_3 \cup \dotsb \cup S_m)
        },
        \\
        J' &= (I\cap S_1) \cup \paren[\big]{
            J \cap (S_2 \cup S_3 \cup \dotsb \cup S_m)
        }
    \end{align*}
    is another supercluster in \(E(X,Y)\).
    The statement follows.
\end{proof}

The remaining terms have \(X\) and \(Y\) in the same supercluster and as for DoC we obtain the bound
\begin{align*}
    \Alignindent
    \abs{
        \trace{B \, \rho_{\Lambda}}
        - \trace{B \, \tilde\rho_{\Lambda}}
    }
    \\&\leq
    \fnfrac{
        Z^0_\Lambda \, \Zt^0_\Lambda
    }{
        Z_\Lambda \, \Zt_\Lambda
    }
    \, \sumstack[lr]{S_0\subset \Lambda\suchthat\\ \textup{\(S_0\) \(\conn\)-connected}\\ X,Y\subset S_0}
    \, \fnfrac{
        Z_{\Lambda\setminus \closure{S_0}} \, \Zt_{\Lambda\setminus \closure{S_0}}
    }{
        Z^0_{\Lambda\setminus \closure{S_0}} \, \Zt^0_{\Lambda\setminus \closure{S_0}}
    }
    \, \sumstack{I_0,J_0\subset \interior{\Lambda}\suchthat\\\mathclap{I_0\cup J_0\cup X\cup Y = S_0}}
    \abs[\bigg]{
        \fnfrac{\trace{B \, T^{\closure{I_0} \cup Y}_{I_0}}}{Z^0_{\closure{I_0} \cup Y}}
        \, \fnfrac{\trace{\Tt^{\closure{J_0}\cup X}_{J_0}}}{\Zt^0_{\closure{J_0}\cup X}}
    }
    + \abs[\bigg]{
        \fnfrac{\trace{B \, \Tt^{\closure{I_0} \cup Y \cup X}_{I_0}}}{\Zt^0_{\closure{I_0} \cup Y \cup X}}
        \, \fnfrac{\trace{T^{\closure{J_0}}_{J_0}}}{Z^0_{\closure{J_0}}}
    }
    .
\end{align*}
We continue the proof analogously to the proof of \cref{thm:DoC}, but adjust for the perturbation~\(W\).
As before, we write for each \(I\), \(\Omega\subset S_0\) and \(O\in \alg_\Omega\),
\begin{equation*}
    \abs[\bigg]{
        \frac{\trace{O \, \Tt^{\closure{I}\cup \Omega \cup X}_{I}}}{\Zt^0_{\closure{I}\cup \Omega \cup X}}
    }
    =
    \abs[\bigg]{
        \frac{
            \trace_{\closure{I}\cup \Omega \cup X}{O \, \Tt^{\closure{I}\cup \Omega\cup X}_{I}}
        }{
            \trace_{\closure{I}\cup X}{\e^{-\beta \Ht^0_{\closure{I}\cup X}}}
            \, \trace_{\Omega\setminus(\closure{I}\cup X)}{\e^{-\beta H^0_{\Omega\setminus(\closure{I}\cup X)}}}
        }
    }
    ,
\end{equation*}
by reducing to a trace over \(\alg_{\closure{I}\cup \Omega\cup X}\).
The numerator is bounded by
\begin{align*}
    \abs{
        \trace_{\closure{I}\cup \Omega\cup X}{O \, \Tt^{\closure{I}\cup \Omega\cup X}_{I}}
    }
    &\leq
    \norm{O}
    \, \norm[\big]{
        \Tt^{\closure{I}\cup X}_{I} \, \e^{-\beta H^0_{\Omega\setminus(\closure{I}\cup X)}}
    }_{\Tr_{\closure{I}\cup \Omega\cup X}}
    \\&\leq
    \norm{O}
    \, \norm[\big]{
        \Tt^{\closure{I}\cup X}_{I}
    }
    \, \norm[\big]{
        \unit_{\closure{I}\cup X}
    }_{\Tr_{\closure{I}\cup X}}
    \, \norm[\big]{
        \e^{-\beta H^0_{\Omega\setminus(\closure{I}\cup X)}}
    }_{\Tr_{\Omega\setminus(\closure{I}\cup X)}}
    \\&\leq
    \norm{O}
    \, (2a)^{\abs{I}}
    \, \e^{\beta \norm{W}}
    \, q^{\abs{\closure{I}\cup X}}
    \, \trace_{\Omega\setminus(\closure{I}\cup X)}{
        \e^{-\beta H^0_{\Omega\setminus(\closure{I}\cup X)}}
    }
    ,
\end{align*}
using \cref{lem:LPPL-TInorm} in the last step.
The second term in the denominator cancels with the last term and for the first term we observe \(
    \Ht^0_{\closure{I}\cup X}
    \leq
    H^0_{\closure{I}\cup X} + \sup_\psi \innerp{\psi}{W \, \psi}
\) and use monotonicity as in the proof of \cref{lem:ratiopartition} to bound
\begin{equation*}
    \trace{\e^{-\beta \Ht^0_{\closure{I}\cup X}}}
    \geq
    \trace{\e^{-\beta H^0_{\closure{I}\cup X}}}
    \, \e^{-\beta \sup_\psi \innerp{\psi}{W \, \psi}}
    \geq
    \e^{-\beta \norm{W}}
    .
\end{equation*}
Moreover, we observe that \cref{lem:ratiopartition} also holds for the tilde variants as long as \(X\subset \closure{S_0}\), as the proof only uses monotonicity to add terms \(v_x \geq 0\).

Putting everything together, we obtain the bound
\begin{align*}
    \Alignindent
    \abs{
        \trace{B \, \rho_{\Lambda}}
        - \trace{B \, \tilde\rho_{\Lambda}}
    }
    \\&\leq
    2 \, \norm{B} \, \e^{2\beta\norm{W}}
    \, \sumstack[l]{S_0\subset \Lambda\suchthat\\ \textup{\(S_0\) \(\conn\)-connected}\\ X,Y\subset S_0}
    \, \sumstack[r]{I_0,J_0\subset \interior{\Lambda}\suchthat\\I_0\cup J_0\cup X\cup Y = S_0}
    \, \paren[\big]{p}^{\abs{I_0}+\abs{J_0}}
    \, q^{\abs{X\setminus I_0}}
    \\&\leq
    2 \, \norm{B} \, \e^{2\beta\norm{W}}
    \, q^{\abs{X}}
    \, \paren*{\fnfrac{1}{2} + \fnfrac{1}{2p}}^{ \abs{X} + \abs{Y}}
    \, \sumstack[lr]{S_0\subset \Lambda\suchthat \\ \textup{\(S_0\) \(\conn\)-connected}\\ X,Y\subset S_0}
    \, \paren[\big]{2p \, (1+p)}^{\abs{S_0}}
    \\&\leq
    2 \, \norm{B} \, \e^{2\beta\norm{W}}
    \, q^{\abs{X}}
    \, \paren*{\frac{1}{2} + \frac{1}{2p}}^{ \abs{X} + \abs{Y}}
    \paren[\big]{2p \, (1+p) \, C}^{\dist{X,Y}/(2 \conn)}
    ,
\end{align*}
where, as before, \(p = 2 \, a \, q^{(2R+1)^D}\) and \(2p \, (1+p) \, C < 1\) for \(a\) small enough.

\statement{Acknowledgments}
This work was funded by the National Science Foundation NSF -- grant DMS 2102842 (AA);
the European Union -- ERC Advanced Grant \enquote{RMTBeyond} No.~101020331 (JH), ERC Consolidator Grant \enquote{ProbQuant}, jointly with the Swiss State Secretariat for Education, Research and Innovation (JH),
ERC Starting Grant \enquote{MathQuantProp} No.~101163620 (ML)%
\footnote{%
    Views and opinions expressed are however those of the authors only and do not necessarily reflect those of the European Union or the European Research Council Executive Agency.
    Neither the European Union nor the granting authority can be held responsible for them.
}%
;
the \foreignlanguage{ngerman}{Deutsche Forschungsgemeinschaft} (DFG, German Research Foundation) -- 470903074 (ML, TW), 465199066 (TW);
the Federal Ministry of Research, Technology and Space (BMFTR) and the Baden-Württemberg Ministry of Science as part of the Excellence Strategy of the German Federal and State Governments (AA, ML).

\statement{Conflict of interest}
The authors have no conflicts to disclose.

\statement{Data availability}
Data sharing is not applicable to this article as no new data were created or analysed in this study.

\printbibliography[heading=bibintoc]

\end{document}